\PassOptionsToPackage{nameinlink}{cleveref}
\documentclass[a4paper, UKenglish, numberwithinsect, thm-restate, cleveref, final]{lipics-v2021}
\usepackage{lineno}
\nolinenumbers

\newcommand{\takeout}[1]{\empty}

\usepackage[utf8]{inputenc}
\usepackage{enumitem}
\usepackage[T1]{fontenc}
\usepackage{stmaryrd}
\usepackage{dsfont}
\usepackage{thmtools}
\usepackage{mathtools}
\usepackage{tikz-cd}
\usetikzlibrary{decorations.pathmorphing}
\usetikzlibrary{calc}
\usepackage{relsize}
\usepackage{scalerel}
\tikzset{shiftarr/.style={
    rounded corners,%
    to path={--([#1]\tikztostart.center)
      -- ([#1]\tikztotarget.center) \tikztonodes
      -- (\tikztotarget)},
  }}
\usepackage[nomargin, marginclue, inline]{fixme}
\fxuselayouts{inline,nomargin}
\fxusetheme{color}
\FXRegisterAuthor{fb}{afb}{FB}
\FXRegisterAuthor{sm}{asm}{SM}
\FXRegisterAuthor{hu}{ahu}{HU}
\usepackage[notref,notcite, color]{showkeys}
\usepackage{xspace}
\usepackage{awesomebox}
\usepackage{microtype}

\usepackage{anyfontsize}
\usepackage{bbm}

\usepackage{cleveref}

\takeout{
\newcommand{\resetCurThmBraces}{%
\gdef\curThmBraceOpen{(}%
\gdef\curThmBraceClose{)}}
\resetCurThmBraces
\newcommand{\removeThmBraces}{%
\gdef\curThmBraceOpen{}%
\gdef\curThmBraceClose{}}

\newenvironment{notheorembrackets}{\removeThmBraces}{\resetCurThmBraces}
\usepackage{etoolbox}
\patchcmd{\thmhead}{(#3)}{\curThmBraceOpen #3\curThmBraceClose }{}{}
}

\def\resettheorembrackets{
  \def\theorembracketopen{(}
  \def\theorembracketclose{)}
}
\resettheorembrackets
\makeatletter
\def\@spopargbegintheorem#1#2#3#4#5{\trivlist
  \item[\hskip\labelsep{#4#1\ #2}]{#4{\theorembracketopen}#3{\theorembracketclose}\@thmcounterend\ }#5}
\makeatother

\newcommand{\resetCurThmBraces}{%
  \gdef\curThmBraceOpen{(}%
  \gdef\curThmBraceClose{)}}
\resetCurThmBraces
\newcommand{\removeThmBraces}{%
  \gdef\curThmBraceOpen{}%
  \gdef\curThmBraceClose{}}
\newenvironment{notheorembrackets}{\removeThmBraces}{\resetCurThmBraces}

\patchcmd{\thmhead}{(#3)}{\curThmBraceOpen #3\curThmBraceClose }{}{}

\newcommand{\pullbackangle}[2][]{\arrow[phantom,to path={
                     -- ($ (\tikztostart)!1cm!#2:([xshift=8cm]\tikztostart) $)
                        node[anchor=west,pos=0.0,rotate=#2,
                        inner xsep = 0]
                        {\begin{tikzpicture}[minimum
                        height=1mm,baseline=0,#1]
    \draw[-] (0,0) -- (.5em,.5em) -- (0,1em);
                        \end{tikzpicture}}}]{}}

\AtBeginDocument{%
  }

\usepackage{enumitem}
\setlist[enumerate,1]{label=(\arabic*),font=\normalfont,align=left,leftmargin=0pt,labelindent=0pt,listparindent=\parindent,labelwidth=0pt,itemindent=!,topsep=2pt,parsep=0pt,itemsep=2pt,start=1}
\setlist[enumerate,2]{label=(\alph*),font=\normalfont,labelindent=*,leftmargin=*,start=1,left=5pt,ref=(\arabic{enumi}\alph*)}
\setlist[itemize]{labelindent=*,leftmargin=*}
\setlist[description]{labelindent=*,leftmargin=*,itemindent=-1 em}

\newcommand{\appendixproof}[2][Proof of]{%
  \subsection*{#1\ \autoref{#2}}%
  \addcontentsline{toc}{subsection}{#1\ \autoref{#2}}%
}
\newcommand{\detailsfor}[1]{%
  \appendixproof[Details for]{#1}%
}

\AtEndPreamble{%

  \theoremstyle{plain}
  \newtheorem{thm}{Theorem}[section]
  \newtheorem{prop}[thm]{Proposition}
  \newtheorem{lem}[thm]{Lemma}
  \newtheorem{cor}[thm]{Corollary}

  \theoremstyle{definition}
  \newtheorem{defn}[thm]{Definition}
  
  \newtheorem{constr}[thm]{Construction}
  \newtheorem{asm}[thm]{Assumption}
  \newtheorem{expl}[thm]{Example}

  \theoremstyle{definition}
  \newtheorem{rem}[thm]{Remark}
  \newtheorem{nota}[thm]{Notation}

  \crefname{thm}{Theorem}{Theorems}
}

\makeatletter
\newcommand{\setlabel}[2]{\edef\@currentlabel{#1}\label{#2}}
\makeatother

\newcommand{\iref}[2]{\autoref{#1}.\ref{#1:#2}}

\newcommand{\f}{\mathsf{f}}
\newcommand{\pow}{\ensuremath{\mathcal{P}}\xspace}
\newcommand{\powf}{\ensuremath{\mathcal{P}_{\f}}\xspace}

\newcommand{\id}{\mathsf{id}}
\newcommand{\Id}{\mathsf{Id}}
\newcommand{\epi}{\twoheadrightarrow}
\newcommand{\mono}{\rightarrowtail}
\newcommand{\incl}{\hookrightarrow}
\newcommand{\darr}{\downarrow}

\newcommand{\aiso}[0]{\ensuremath{\mathbf{a}}\xspace}

\newcommand{\cat}[1]{\ensuremath{\mathbf{#1}}\xspace}
\newcommand{\A}{\cat{A}}
\newcommand{\names}{\ensuremath{\mathbb{A}}\xspace}
\newcommand{\C}{\cat{C}}
\newcommand{\V}{\cat{V}}
\newcommand{\ctx}{\cat{Ctx}}

\newcommand{\D}{\cat{D}}
\newcommand{\E}{\cat{E}}
\newcommand{\F}{\ensuremath{\mathbb{F}}\xspace}
\newcommand{\I}{\ensuremath{\mathbb{I}}\xspace}
\newcommand{\Idcat}{\ensuremath{\mathbb{J}}\xspace}
\newcommand{\Inc}{\ensuremath{\Idcat_{\subseteq}\xspace}}
\renewcommand{\S}{\ensuremath{\mathbb{S}}\xspace}
\newcommand{\B}{\ensuremath{\mathbb{B}}\xspace}

\newcommand{\Set}{\cat{Set}}

\newcommand{\xto}{\xrightarrow}

\newcommand{\Nom}{\cat{Nom}}
\newcommand{\Nompres}{\ensuremath{{\cat{Nom}_{=}}}\xspace}
\newcommand{\Suppset}{\cat{SuppSet}}
\newcommand{\Suppsetpres}{\ensuremath{\cat{SuppSet}_=}}
\newcommand{\Ren}{\cat{Ren}}
\newcommand{\Relev}{\cat{Relev}}
\newcommand{\Relevpres}{\ensuremath{\cat{Relev}_{=}}\xspace}

\newcommand{\fs}{\mathsf{fs}}

\newcommand{\tr}[2]{({#1}\,\,{#2})}
\DeclareMathOperator{\supp}{\mathsf{supp}}
\DeclareMathOperator{\Perm}{\mathsf{Perm}}
\DeclareMathOperator{\rens}{Ren}
\DeclareMathOperator{\orb}{orb}

\DeclareMathOperator{\colim}{colim}

\newcommand{\op}{\mathsf{op}}

\DeclareMathOperator{\Sh}{\mathbf{Sh}}

\newcommand{\PSh}[1]{[{#1}, \Set]}
\newcommand{\PShpb}[1]{[{#1}, \Set]_{\cap}}
\newcommand{\PShpi}[1]{[{#1}, \Set]_{\mathsf{pi}}}
\newcommand{\y}{\mathbf{y}}
\newcommand{\sh}{I^*}
\newcommand{\unsh}{I_*}
\DeclareMathOperator{\Nat}{\mathsf{Nat}}
\DeclareMathOperator{\Lan}{\mathsf{Lan}}

\DeclareMathOperator{\strabs}{[{\A}]_=}
\newcommand{\sub}{\mathbin{\diamond}}

\DeclareMathOperator{\curry}{\mathsf{curry}}
\DeclareMathOperator{\uncurry}{\mathsf{uncurry}}

\newcommand{\quot}{\!\mathrel{/\!}}
\newcommand{\bigast}{\ast}
\newcommand{\bigst}{\raisebox{-0.75pt}{\rotatebox{30}{\scalebox{1}{$\bigast$}}}}
\newcommand{\wand}{\mathbin{\relbar\mkern-10mu{\bigst}}}

\newcommand{\dwand}{\mathbin{\relbar\mkern-7mu\diamond}}
\newcommand{\uswand}{\mathbin{\relbar\mkern-6.5mu\raisebox{1pt}{\scalebox{0.7}{\(\otimes\)}}}}

\newcommand{\tl}{\mathbin{\triangleleft}}

\newcommand{\overbar}[1]{\mkern 1.5mu\overline{\mkern-1.5mu#1\mkern-1.5mu}\mkern 1.5mu}
\newcommand*\cocolon{%
  \nobreak
  \mskip6mu plus1mu
  \mathpunct{}%
  \nonscript
  \mkern-\thinmuskip
  {:}%
  \mskip2mu
  \relax
}

\title{A Unified Treatment of the Substitution Tensor for \\ Presheaves, Nominal Sets, Renaming Sets, \\ and so on}
\titlerunning{A Unified Treatment of the Substitution Tensor}
\author{Fabian Lenke}{Friedrich-Alexander-Universität Erlangen-Nürnberg, Germany\and\url{https://www8.cs.fau.de/people/fabian-lenke/}}{fabian.birkmann@fau.de}{https://orcid.org/0000-0001-5890-9485}{}
\author{Stefan Milius}{Friedrich-Alexander-Universität Erlangen-Nürnberg, Germany\and\url{https://www.stefan-milius.eu}}{stefan.milius@fau.de}{https://orcid.org/0000-0002-2021-1644}{}
\author{Henning Urbat}{Friedrich-Alexander-Universität Erlangen-Nürnberg, Germany\and\url{https://www8.cs.fau.de/people/henning-urbat/}}{henning.urbat@fau.de}{https://orcid.org/0000-0002-3265-7168}{}
\authorrunning{F.~Lenke, S.~Milius, and H.~Urbat}

\Copyright{Fabian Lenke, Stefan Milius and Henning Urbat}

\relatedversiondetails{Full version}{https://arxiv.org/abs/2602.11907}
\funding{Fabian Lenke is supported by Deutsche Forschungsgemeinschaft (DFG, German Research Foundation) -- project number 419850228.
  Henning Urbat are supported by Deutsche Forschungsgemeinschaft (DFG, German Research Foundation) -- project number 569130867.
  Stefan Milius is supported by Deutsche Forschungsgemeinschaft (DFG, German Research Foundation) -- project number 517924115.}

\numberwithin{equation}{section}

\makeatletter
\hypersetup{
  final,
  hidelinks,
  pdftitle={\@title},
  pdfauthor={Lenke, Milius, Urbat},
}
\makeatother

\begin{document}

\maketitle

\begin{abstract}
  Presheaves and nominal sets provide alternative abstract models of sets of syntactic objects with free and bound variables, such as $\lambda$-terms.
  One distinguishing feature of the presheaf-based perspective is its elegant syntax-free characterization of substitution using a closed monoidal structure.
  In this paper, we introduce a corresponding closed monoidal structure on nominal sets, in the spirit of Fiore~et~al.'s substitution tensor for presheaves over finite sets.
  To this end, we present a general method to derive a closed monoidal structure on a category from a given action of a monoidal category on that category.
m  We demonstrate that this method not only uniformly recovers known substitution tensors for various kinds of presheaf categories but also yields notions of substitution tensor for nominal sets and their relatives, such as renaming sets. In the process, we shed new light 
  on different incarnations of nominal sets and (pre-)sheaf categories and establish a number of correspondences between them.
\end{abstract}

\section{Introduction}

Substitution is a ubiquitous operation in the theory of computation. It naturally appears whenever syntactic objects (terms, formulas, programs) that involve variables are manipulated, such as in algebra, logic, type theory, and programming language theory.
For instance, in one of the fundamental models of computation, the $\lambda$-calculus, the essence of the operational semantics lies in the $\beta$-reduction rule
\( (\lambda x.\, p) \, q \; \to_\beta \; p[q/x]  \)
expressing that ``to apply a function $\lambda x.\, p$ to the input $q$, substitute~$q$ for every free occurrence of the variable $x$ in the term $p$''. Even in such simple models, substitution tends to be rather subtle: the presence of free and bound variables and $\alpha$-renamings requires careful bookkeeping to avoid capture. The situation becomes increasingly complex in more advanced settings, including for example substructural systems such as the linear $\lambda$-calculus which impose restrictions on the use of variables (resources), or systems involving parametric or higher-order types. 

To navigate these difficulties, the investigation of {abstract}, syntax-free mathematical models of substitution has spurred the interest of researchers for a long time. At LICS'99, two independent seminal works in this direction were presented by Fiore, Plotkin, and Turi~\cite{ftp99} and by Gabbay and Pitts:
\begin{enumerate}
\item The former~\cite{ftp99} studied \emph{presheaves} as a model of abstract syntax and variable binding. Their original work
  considered the category $\PSh{\F}$ of presheaves over finite sets (cartesian contexts)
  and captures the syntax of languages like the basic untyped $\lambda$-calculus. The idea is to interpret the collection of $\lambda$-terms 
  as a functor 
  \( \Lambda\colon \F \to \Set \)
  that sends a finite context of variables to the set of $\lambda$-terms (modulo $\alpha$-equivalence) in that context. Operations like $\lambda$-abstraction and substitution are then captured at this abstract level of presheaves via initial algebra semantics. The presheaf-based approach was
  subsequently extended in a series of papers in two orthogonal directions. On the one hand,
  it turned out to also apply smoothly to \emph{substructural} abstract syntax, which
  amounts to replacing the category $\F$ of cartesian contexts with the categories $\I$,
  $\S$, $\B$ of finite sets and injections, surjections, and bijections, corresponding to
  affine, relevant, and linear contexts, respectively~\cite{t00,fr25}. On the other hand,
  various authors have studied presheaf models of (dependently or parametrically)
    \emph{typed} languages, using more complex categories of typed contexts in lieu of $\F$~\cite{f08,hamana11}. Both directions highlight the power and flexibility of the presheaf perspective to model various, even rather complex forms of abstract syntax.
\item \cite{gp99} introduced a conceptually very different approach to abstract syntax and variable
  binding using \emph{permutations} of names, modeled by a suitable group actions. Here the collection of $\lambda$-terms is viewed as a \emph{set} $\Lambda$ equipped with an action 
  \( \Perm \names \times \Lambda \to \Lambda, (\pi,p) \mapsto \pi\cdot p \),
of the group of permutations of (variable) names that performs capture-avoiding renamings; e.g.\ $(x\, y)\cdot (\lambda x.\, x\, y) = \lambda z.\, z\, x$. Originally based on Fraenkel-Mostowski (FM) set theory~\cite{g01}, their approach
  ultimately developed into the theory of \emph{nominal sets}~\cite{p13}, which has found
  broad applications in many areas of computer science, among them programming languages~
  \cite{bbkl12,s06}, logic~\cite{l24,pit03,gc04}, algebra~\cite{gm09,kp10}, and automata
  theory~\cite{bkl14,skmw17,hmsu25,uhms21}.
\end{enumerate}

Even though the topics and intentions of the two original papers~\cite{ftp99,gp99} are very
similar (even their titles are), the subsequent lines of research they have spawned were increasingly unrelated. In particular, nominal sets were largely neglected as categorical models of substitution via monoids in favor of presheaves. One may ask why, given the perks of nominal sets: they are in many respects simpler and more well-behaved structures than presheaves~\cite{gh08}, which may be the key reason for their adoption outside the category theory community (e.g.\ in automata theory). It is often more intuitive and natural, and technically easier, to think of the collection of $\lambda$-terms as a (nominal) \emph{set}, as opposed to a \emph{functor}. In fact, much of the theory of nominal sets can be developed in the elementary language of group actions without reference to categorical concepts such as (pre)sheaves at all.

However, one can easily identify one specific reason for the focus on presheaves: While name abstraction is well-understood in both nominal sets and presheaves,
the \emph{substitution tensor}, which models free simultaneous substitution in the presheaf setting via a closed monoidal structure~\cite{ftp99} still has, over 25 years later, found no analogue in the nominal world.
Power~\cite{jp07} and Power and Tanaka~\cite{pt08} conjectured that this structure could be introduced to nominal
sets via a unifying $2$-categorical approach using pseudo-distributive laws over
pseudo-monads, but this direction has not been pursued further in subsequent work. This gap is even more remarkable given the close connection between nominal sets and presheaves: there exist multiple characterizations of nominal sets as well-behaved presheaves (namely intersection-preserving presheaves, or sheaves for a suitable topology), and conversely, 
presheaf categories for which the substitution tensor exists admit corresponding nominal versions. Most notably, Gabbay and Hofmann~\cite{gh08} introduced \emph{renaming sets} (which allow for non-injective renamings rather than just permutations) as a nominal counterpart of the presheaf category $\PSh{\F}$ of~\cite{ftp99}.

\subparagraph*{Contributions} In this work, we introduce a general method for deriving substitution-like closed monoidal structures that covers both the presheaf and the nominal setting. The core idea is to \emph{freely} generate the substitution tensor on a category $\A$ from a \emph{left action} \(\tl\colon \V\times\A\to\A\) of a monoidal category $\V$ on $\A$ that specifies what a substitution is in the given setting. We give a general criterion (\Cref{thm:gen-mon,thm:gen-mon-clo}) for the substitution tensor to yield a closed monoidal structure.

We first instantiate our method to capture the substitution tensors in presheaf categories $\PSh{\ctx}$ for various forms of (untyped) contexts $\ctx$: exchange ($\B$), weakening + exchange, ($\I$) contraction + exchange ($\S$), and cartesian ($\F$). The respective substitution tensors were introduced separately and from scratch in earlier work. We identify a condition (called \emph{contextuality}) on the category $\ctx$ that enables us to construct a closed monoidal structure in $\PSh{\ctx}$ and captures all the four specific cases uniformly.

As the primary application of our theory, we derive a closed substitution tensor for nominal sets and renaming sets, and give elementary nominal descriptions of both the tensors and its internal hom, showing that they are considerably simpler than their corresponding presheaf-based versions.
To further stimulate the transfer of concepts between presheaves and nominal sets, we establish a number of correspondences with presheaf types not considered before:
presheaves over $\B$ are equivalent to the category of nominal sets and support-preserving functions,
and presheaves over $\S$ are equivalent to the category of renaming sets where renamings respect least supports. Adding relations between the various presheaf categories and types of nominal sets, we arrive at a more complete map of the nominal landscape.

\subparagraph*{Outline} In \Cref{sec:substitution-intro} we briefly recall the concrete instance of the substitution tensor on presheaves over finite sets~\cite{ftp99}.
In \Cref{sec:categ-presh} we recall some required concepts from category theory such as (co-)ends, Kan extensions, and constructions on presheaves.
\Cref{sec:mon-struct-from-act} establishes our general result for deriving a monoidal structure from a left action, together with a criterion for its closedness.
We instantiate this result to presheaf categories in \Cref{sec:psh-substitution} to uniformly derive their substitution structures.
The corresponding substitution structures for nominal sets and renaming sets are constructed in \Cref{sec:subst-nomin-sets}.
In \Cref{sec:taxon-nomin-sets} we study the relationships between presheaf categories and different incarnations of nominal sets.

\subparagraph*{Related Work}
The substitution tensor has a long history, going back to Kelly~\cite{k05o} for operads and Johnstone and Wraith~\cite{jw78} for their work on (finitary) algebraic theories in toposes.
Fiore, Plotkin and Turi~\cite{ftp99} then connected it to algebras with binders, and their approach has been generalized to increasingly complex settings~(e.g.~\cite{fr25,f08,fs22,f07}.

In nominal sets, single-variable substitution via structural recursion was developed in~\cite[Example 5.7]{gp99}, and investigated via nominal algebra in~\cite{gm08,g16}.
In FM-sets, different possible substitution actions were thoroughly studied in~\cite{g09}.
Notably, the characterisation from~\cite{ftp99} of substitution algebras as monoids is missing for nominal sets.

A possible unification of the two approaches was outlined by Power~\cite{jp07}, using an extension of the 2-categorical axiomatization of substitution developed by Power and Tanaka~\cite{t04,pt08}.

\subparagraph*{Acknowledgements}
We thank Nathanael Arkor and the reviewers for the various pointers to the literature, insightful comments and helpful suggestions and simplifications.

\section{Substitution, Abstractly}\label{sec:substitution-intro}
To provide some intuition and motivation for the general theory of substitution developed in our paper, we recall one concrete categorical setting for modeling substitution:
the category $\PSh{\F}$ of covariant presheaves over the category $\F$ of finite sets and functions. This is the setting originally considered by Fiore et al.~\cite{ftp99}.

\subparagraph*{Models.} Intuitively, we think of an object $A=\{a_1,\ldots,a_n\}$ of $\F$ as a context of $n$ untyped variables, and of a presheaf $X\in \PSh{\F}$ as a map assigning to every context $A$ a set $XA$ of terms in that context, built over some syntax. Given a map $f\colon A\to A'$ in $\F$, the map $Xf\colon XA\to XA'$ sends $t\in XA$ to the term $Xf(t)$ obtained by renaming all (free) variables in $t$ according to $f$.
For example, the presheaf $V^I A = A^I$ represents $I$-indexed families $(a_i \in A)_{i \in I}$ of variables, with special case $V A = V^1 A  = A$, the \emph{presheaf of variables}.
Sets of terms that involve variable binding and $\alpha$-equivalence can also be naturally presented as presheaves.
For example, we can form the presheaf  $\Lambda$ sending a context $A$ to the set $\Lambda A$ of $\lambda$-terms (modulo $\alpha$-equivalence) with free variables from $A$.

\subparagraph*{Substitution.}
A \emph{substitution} $\sigma$ specifies for each variable $a$ from a context $A \in \F$ a term $t_a$ from a model $X$ to be substituted for~$a$. Thus, $\sigma$ is an element of the $A$-fold power of $X$, denoted by
\(A \tl X = X^A\)
We can model the process of applying a substitution~$\sigma$ to a term $t$ of $X$ abstractly in two steps.

\begin{enumerate}
\item Given presheaves $X, Y \in \PSh{\F}$ we can build the presheaf $X \sub Y$ of ``terms of $Y$ freely substituted into terms of $X$''.
It consists of terms $t \in XA$ with variables from $A$ together with a substitution $\sigma \in A \tl Y$.
        However, one has to be careful not to add too many terms: for example, for the presheaf $XA = \pow_\f A$ (where $\pow_\f$ is finite power set), if we formally substitute $[a \mapsto t_a, b \mapsto t_b]$ into $\{a, b\} \in X\{a, b\}$, we obtain the same result ``$\{t_a, t_b\}$'' as \[\{a, b\}[a \mapsto t_a, b \mapsto t_b], \qquad \{a, b\}[a \mapsto t_b, b \mapsto t_a] \qquad\text{or}\qquad \{a, b, c\}[a \mapsto t_a, b \mapsto t_b, c \mapsto t_b].\]
This naturally suggests the following definition:
\begin{equation}
  \label{eq:F-sub}
  X \sub Y = \big( \coprod_{A \in \F} XA \times (A \tl Y) \big) \quot \; \sim, \qquad  \text{ where }\qquad (Xf(t), \sigma) \sim (t, \sigma \cdot f)
\end{equation}
generates the equivalence for $C \in \F$, $t \in XA$, $\sigma \in (B \tl Y)C$, and $f \colon A \rightarrow B$, and where $\sigma \cdot f$ is the ``rendering'' of the substitution $\sigma$ under $f$. Note that the substitution tensor $\sub$ is defined using $\tl$. We can also recover $\tl$ from $\sub$ via \[A \tl Y \cong A \tl (V \sub Y) \cong (A \tl V) \sub Y.\]
\item The actual process of substitution (requiring $Y = X$) is modeled by specifying a ``bind'' operation $X \sub X \rightarrow X$, analogous to that from functional programming languages such as Haskell:
it sends a pair $[t, \sigma]$ to the term $t\sigma$ obtained by applying the substitution $\sigma$ to the free variables of $t$. For the presheaf of $\lambda$-terms for example, this bind operation performs the usual capture-avoiding substitution, e.g.\ the pair $[\lambda x.\, x\, y, [y\mapsto x\,y]]_\sim$ is sent to $\lambda z.\, z\,y$.
\end{enumerate}

The substitution tensor $\sub$ yields a (non-symmetric) \emph{monoidal structure} on $\PSh{\F}$ whose unit is the presheaf $V$ of variables. Moreover, this structure is \emph{right-closed}: there is a natural isomorphism \(\Nat(X \sub Y, Z) \cong \Nat(X, Y \dwand Z) \) for all \( X,Y,Z\in \PSh{\F}\), where the internal hom is given by \((Y \dwand Z)A = \Nat(Y^{A}, Z)\), the \emph{clone of operations} from $Y$ to $Z$. Intuitively, a natural transformation $f\colon Y^A\to Z$ describes, for a fixed term $t$ over free variables in $A$, the outcome $f(\sigma)$ of applying a substitution $\sigma$ to $t$.

\section{Categorical Preliminaries}
\label{sec:categ-presh}

We next work towards our goal of generalizing the construction of the closed monoidal substitution structure sketched in the previous section to the level of abstract categories, with categories of presheaves or nominal sets as concrete instances. This requires some machinery from category theory, notably (co-)ends, Kan extensions, and Day convolution, which we recall next. Readers should be familiar with basic categorical concepts such as (co-)limits, monads, and monoidal categories~\cite{mac-71}.

\begin{nota}
  \begin{enumerate}
    \item For a small category \C, we write \(\PSh{\C} \) for the category of covariant presheaves on $\C$.
          The Yoneda embedding is denoted by
 \(           \y_\C \colon \C^{\op} \rightarrow \PSh{\C}\) mapping \(C \mapsto  \C(C, -)\).
          We drop its subscript if the category $\C$ is clear.
    \item Following~\cite{MacLaneMoerdijk1992}, given a presheaf \(F\), a morphism \(f \colon C \to D\), and an element \(x \in FC\), we write \(f \cdot x\) for \(F(f)(x)\), and \(x \cdot f\) if \(F\) is contravariant (if $F$ is bivariant we use both). Note that on hom-functors this notation agrees with morphism composition, that is, for the bivariant functor \( \C(-, -) \colon \C^\op \times \C \rightarrow \Set\) we have for $f \colon C \rightarrow D$ and $g \colon D \rightarrow E$ that \[\C(C, -)(g)(f) = g \cdot f = \C(-, E)(f)(g).\]
          If \(\alpha \colon F \rightarrow G\) is a natural transformation, we omit the object subscripts of components, so naturality reads as \(\alpha(f \cdot x) = f \cdot \alpha(x)\).
  \end{enumerate}
\end{nota}

\subparagraph*{Coends and Copowers.}
(Co-)ends are a variant of (co-)limits, formed over diagrams of type $H\colon I^\op \times I \rightarrow \C$ for some category~$I$.
A \emph{cowedge} for $H$ consists of an object $W \in \C$ and a family of morphisms $w_i \colon  H(i, i) \rightarrow W$ ($i\in I$) such that
 \[w_i \cdot H(f, i) = w_j \cdot H(j, f) \colon H(j, i) \rightarrow W \] for all $f \colon i \rightarrow j$ in $I$.
A \emph{coend} for $H$ is a cowedge \[\kappa_i \colon H(i, i) \rightarrow  \int^{i \in I}H(i, i) \]
such that every cowedge $W$ for $H$ factorizes through $\kappa_i$ via some unique  $\int^{i \in I}H(i, i) \rightarrow W$.
The notion of an \emph{end}, denoted $\int_{i \in I} H(i, i)$, is dual. 
We are mostly concerned with (co-)ends in \Set, which are formed similarly to (co-)limits: Given \(H \colon I^{\op} \times I \rightarrow \Set\) we have
\begin{align*}
  \int_{i \in I} H(i, i) ~&\cong~\{(x_{i})_{i \in I} \in \prod_{i \in I} H(i, i) \mid \forall i \xrightarrow{f} j \colon f \cdot x_{i} = x_{j} \cdot f \}\\
  \int^{i \in I} H(i, i) ~&\cong~  \coprod_{i \in I} H(i, i) / \sim, \qquad \text{where \(x \cdot f \sim f \cdot x\) for \(f \colon i \rightarrow j\) and \(x \in H(j, i)\)}
\end{align*}
generates the equivalence relation $\sim$.

A category \(\C\) is \emph{copowered} if for every object \(C \in \C\) the covariant representable \(\C(C,-)\) has a left adjoint \(- \cdot C \colon \Set \rightarrow \C \). The object $S \cdot C$ is the \emph{copower} of $C$ by the set $S$. This makes \(\cdot\) a bifunctor \(\Set \times \C \rightarrow \C\). If \(\C\) has coproducts, then \(S \cdot C \cong \coprod_{s \in S} C\).
In particular,
copowers in \Set are simply products: \(S \cdot C = S \times C\) for \(C \in \Set\),
and copowers in presheaves are taken point-wise.

Both coends and copowers can be expressed as colimits of ordinary diagrams; in particular, if a category is cocomplete it has all coends and copowers.

\subparagraph*{The Yoneda Lemma.} The most important result for presheaves is the \emph{Yoneda lemma}:
\begin{lem}[Yoneda]
  For every small category \(\C\) we have the following isomorphisms, natural in \(F \in [\C, \Set]\) and \(C \in \C\):
  \fbnote{TODO correct labels}
  \begin{equation}\label{eq:y}
    \int_{D \in \C} \Set(\y C(D), FD)  \quad\overset{(1)}\cong\quad \Nat(\y C, F)
                                      \quad\overset{(2)}\cong\quad FC
                                      \quad\overset{(3)}\cong\quad \int^{D \in \C} \y D(C) \times FD.
  \end{equation}
\end{lem}

\begin{rem}
  \Cref{eq:y}.(3) is sometimes called the \emph{Co-Yoneda lemma} or \emph{density}, since it implies, together with the fact that coends can be expressed as colimits, density of the Yoneda embedding: every presheaf is a colimit of representables.
\end{rem}

\subparagraph*{Kan Extensions}
 The concept of \emph{Kan extension} laxly solves the factorization problem for functors. Given functors \(\D \xleftarrow{K} \C \xrightarrow{F} \E\),
  the \emph{left Kan extension of $F$ along $K$} consists of a functor $\Lan_{K} F \colon \D \rightarrow \E$ with a natural transformation $\eta \colon F \rightarrow \Lan_{K}F \cdot K $ such that every natural transformation $\alpha \colon F \rightarrow GK$ factorizes uniquely as 
  \[\alpha = (\, F \xto{\eta} (\Lan_K F)K \xto{\hat\alpha K} GK\,) \qquad \text{for some $\hat{\alpha} \colon \Lan_{K}F \rightarrow G$.} \]
  A functor \(H \colon \E \rightarrow \E'\) \emph{preserves} the Kan extension \(\Lan_{K}F\) if \(H \cdot \Lan_{K}F \cong \Lan_{K}HF\).

We recall some important properties of Kan extensions, see e.g.~\cite[Chapter X]{mac-71} for proofs.
\begin{prop}\label{prop:kan-props}
  \begin{enumerate}
    \item\label{prop:kan-props:exist} If \C is small and \E is cocomplete, then for every $K$ and $F$ as above, the left Kan extension $\Lan_K F$ is given by \eqref{eq:kan-ext}, and there is an adjunction \eqref{eq:lan-adj},\vspace{-8pt}\\
          \begin{minipage}{0.45\linewidth}
            \begin{equation}
              \label{eq:kan-ext}
              \Lan_{K}F \cong \int^{C \in \C}\D(KC, -) \cdot FC
            \end{equation}
          \end{minipage}
          \hfill%
          \begin{minipage}{0.45\linewidth}
            \begin{equation}
              \label{eq:lan-adj}
              \Lan_K\dashv K^* \colon  [\D, \E] \rightarrow [\C, \E]
            \end{equation}
          \end{minipage}\\
          where $K^*$ is precomposition with $K$.
    \item\label{prop:kan-props:adj} Left adjoints preserve left Kan extensions.
    \item\label{prop:kan-props:ff} If $K$ is fully faithful, then $\Lan_K F$ extends $F$: we have \(K^* (\Lan_K F) = (\Lan_K F)K \cong F\).
  \end{enumerate}
\end{prop}

\begin{rem}\label{rem:psh-functor}
\begin{enumerate} \item   Left Kan extensions along the Yoneda embedding are particularly simple.
  For a small category \(\C\) and a complete category~\(\E\), the left Kan extension  of a functor \(F \colon \C^{\op} \rightarrow \E\) along $\y\colon \C^\op\to [\C, \Set]$
  simplifies by the Yoneda lemma to
  \[(\Lan_{\y}F)X =  \int^{C} XC \cdot FC.\]
  The functor $\Lan_\y F$ has a right adjoint, the \emph{nerve functor} 
  \begin{align}
    \label{eq:nerve}
    \nu_F \colon \E \rightarrow [\C, \Set], \qquad \nu_F E = \E(F(-), E).
  \end{align}
 \item  The functor $F$ is \emph{dense} if the comonad on $\E$  induced by  the adjunction $\Lan_\y F \dashv \nu_F$  is the identity comonad, that is,
  \(\Lan_\y F \cdot \nu_F \cong \Id\).
  This is equivalent to $\nu_F$ being fully faithful, or more concretely, that every $E \in \E$ is the canonical colimit of the diagram
        \((F \darr \E) \rightarrow \E \) given by \((f \colon FC \rightarrow E) \mapsto FC\).
        Dense functors are cancellable~\cite[Theorem 5.13]{k05}: if $K \cong JP$ is dense and $J$ is fully faithful then both $P$ and $J$ are dense.
  \item The construction of presheaf categories is (pseudo-)func\-to\-rial: for a functor  $G\colon \C\to \D$ of small categories, we put  \(\E = \PSh \D\) and
        \(F = \y_{\D} \cdot G^{\op} \colon \C^\op \rightarrow \D^\op \rightarrow \PSh \D \), then we have $\Lan_{\y_{\C}}F \cong \Lan_G$:
  \begin{equation}
    \label{eq:psh}
    \Lan_G \colon \PSh \C \rightarrow \PSh \D, \qquad X \mapsto \int^{C} XC \cdot F C.
  \end{equation}
  Note that \(\Lan_G\) \emph{extends} \(G\), viz.\  \(\Lan_G(\y_{\C}C) \cong \y_{\D}(GC)\), and has $G^*$ as a right adjoint.
\end{enumerate}
\end{rem}

\subparagraph*{Day Convolution.}
It is well-known that every monoid $M$ extends uniquely to a join-bilinear monoid on its powerset $\pow M$ whose multiplication agrees with $M$ on singletons.
This construction has a categorical generalization~\cite{day70}: for a small monoidal category \((\C, \otimes, I)\), the category \(\PSh \C\)
can be equipped with a unique monoidal structure \(*\) preserving colimits in both variables, called the \emph{Day convolution}~\cite{day70}. A concrete formula is therefore by
\begin{align*}
  X * Y \cong \int^{C, C' \in \C} X C \cdot X' C' \cdot \y (C \otimes C')
\end{align*}
with unit \(\y I \in \PSh \C\).
Note that the presheaf \(X * Y\) is given by the left Kan extension of \(\C \times \C \xrightarrow{X \times Y} \Set \times \Set \xrightarrow{\times} \Set\) along \(\C \times \C \xrightarrow{\otimes} \C\).

The monoidal category $(\PSh \C, *, \y I)$ is biclosed, with internal homs given by  \[(X \wand Y)(C) \;=\; \Nat(X,Y(C \otimes -)) \;\cong\; \Nat(\y C * X, Y).\]

 Equipping $\PSh \C$ with the monoidal structure of Day convolution makes \(\y \colon \C^{\op} \to \PSh \C\) a strong monoidal functor. Moreover,~\(*\) preserves colimits in each argument (due to biclosure). In fact, Day convolution is universal for these two properties~\cite{bk86}:
\begin{thm}\label{thm:day-conv-pres}
  Let \((\C, \otimes)\) be a small monoidal category and let \((\D, \otimes)\) be a cocomplete monoidal category whose tensor preserves colimits in each argument.
  Every strong monoidal functor $F \colon  \C^{\op} \rightarrow \D$ extends, via left Kan extension along $\y$, to a unique colimit-preserving strong monoidal functor $\hat{F} \colon \PSh \C \rightarrow \D$ with $\hat{F} \cdot \y \cong F$.
\end{thm}

\section{Monoidal Structures from Actions}
\label{sec:mon-struct-from-act}

In this section we present  our general method to construct a closed monoidal structure on a category from a suitable left action on that category. We will show in subsequent sections that this method instantiates to both categories of presheaves and nominal sets and yields the construction of their respective substitution tensors.

\begin{defn}[Left action~\cite{jk01}]
  A \emph{(left) action} of a monoidal category $(\V, \otimes, I)$ on a category $\A$ is given by a functor \(\tl \colon \V \times \A \rightarrow \A\) with isomorphisms $I \tl A \cong A$ and $V \tl (U \tl A) \cong (V \otimes U) \tl A$ satisfying coherence conditions similar to monoidal categories.
\end{defn}

\begin{expl}[$\PSh \F$]\label{ex:setf-action}
As outlined in \Cref{sec:substitution-intro}, the intention is that $\tl$ models a notion of substitution. For instance, we have seen that substitution in $\A=\PSh \F$ is captured by the left action
\( \tl\colon \F^\op \times \PSh \F \to \PSh \F\) defined by \(V\tl X = X^V\).
\end{expl}

Left actions, and more generally functors $F\colon \V\times \A\to\A$, can be extended to bifunctors on $\A$ via the following technique, mirroring the construction of the substitution tensor \eqref{eq:F-sub}.

\begin{constr}[$J$-extension]\label{constr:ext}
  Let $J \colon \V \rightarrow \A$ be a functor from a small into a cocomplete category. The \emph{$J$-extension} of a bifunctor $F \colon \V \times \A \rightarrow \A$ is the bifunctor $F_J$ given by
  \begin{align}
    \label{eq:ext}
\Lan_{J \times 1}F \colon \A \times \A \rightarrow \A, \qquad
    (A, B) \mapsto \int^{V \in \V} \A(JV, A) \cdot F(V, B).
  \end{align}
\end{constr}

\begin{expl}[$\PSh \F$]
In the setting of \Cref{ex:setf-action}, the functor~$J$ is given by the Yoneda embedding $\y\colon \F^\op\to \PSh{ \F }$, and the $\y$-extension of $\tl$ yields  precisely the substitution tensor of $\PSh\F$, that is,
\( X\sub Y \;=\; X\tl_\y Y \).
\end{expl}
We now investigate when the extension of a left action gives rise to a closed monoidal structure. The key ingredient is the notion of a \emph{well-behaved functor}~\cite[Def.~4.1]{acu14}:

\begin{defn}[Well-behaved functor]
  A functor $J \colon \V \rightarrow \A$ is \emph{well-behaved for $F \colon \V \rightarrow \A$} if is fully faithful, dense, and its nerve preserves the left Kan extension $\Lan_J F$:
  \begin{align}
    \label{eq:wb}
    \Lan_J (\A(JV, F-)) A \cong \A(JV, \Lan_J(F)A) \qquad \text{for all $V \in \V$ and $A \in \A$.}
  \end{align}
  The functor $J$ is \emph{well-behaved} if it is well-behaved for all $F \colon \V \rightarrow \A$.
\end{defn}

\begin{expl}[Well-behaved functors]\label{expl:wb}
  \begin{enumerate}
    \item\label{expl:wb:y} For every small category $\V$, the Yoneda embedding $\y\colon \V^\op \rightarrow \PSh \V$ is well-behaved because the nerve $\nu_\y \cong \Id$ is trivial by Yoneda.
    \item\label{expl:wb:fset} The inclusion $V \colon \F \incl \Set$ is well-behaved~\cite{acu14}.
  \end{enumerate}
\end{expl}

Reflective subcategories of presheafs closed under colimits weighted by the subcategory (called \emph{self-cocomplete} in \cite[Section 8]{s17}) yield a common source of well-behaved functors.
\begin{lem}\label{lem:wb-psh-subcat}
  Let $\iota \colon \A \incl \PSh \V$ be a full reflective subcategory containing all representables.
  Then the corestriction $\y' \colon \V^\op \rightarrow \A$ of the Yoneda embedding is well-behaved for $F \colon \V^\op \rightarrow \A$ iff $\Lan_{\iota F}\colon \PSh \V \rightarrow \PSh \V$ (co-)restricts to $\A$.
\end{lem}
\begin{proof}
  That $\Lan_\y (\iota F)$ corestricts to $\A$ is equivalent to commutativity of the following diagram:
  \begin{equation}
    \label{eq:wb-psh-subcat}
    \begin{tikzcd}
      \V^\op \dar{\y'} \ar[shiftarr={xshift=-35pt}]{dd}{\y} \drar{F} & \\
      \A \dar{\iota} \rar[swap]{\Lan_{\y'} F} & \A \dar{\iota} \\
      \PSh \V \rar{\Lan_\y(\iota F)} & \PSh \V
    \end{tikzcd}
  \end{equation}
  Since the $\y'$-nerve is isomorphic to $\iota$,
  the corestriction $\y'$ of the Yoneda embedding is well-behaved with respect to $F \colon \V^\op \rightarrow \A$ iff for all $A \in \V$ and all $C \in \A$
  \begin{equation}
    \label{eq:psh-sub-lan}
    \Lan_{\y'}(\A(\y' A, F))C \cong \int^B \iota(C)B \cdot \iota(FB)A
  \end{equation}
  is isomorphic to
  \begin{equation}
    \label{eq:psh-sub-nerve}
    \A(\y' A, \Lan_{\y'}(F)C) \cong \iota(\int^{B}\iota(C)B \cdot FB)A,
  \end{equation}
  which means precisely that \eqref{eq:wb-psh-subcat} commutes.

  Recall for a diagram $X_i \in \A$ we have that  $\colim_i \iota X_i$ is again in $\A$
  (more precisely: $\iota L \colim_i \iota X_i \cong \colim_i \iota X_i$)
  iff the isomorphism $\colim_i \iota X_i \cong \iota \colim_i X_i$ holds, since
  either of these isomorphisms extends the following chain
  \begin{align*}
   \iota L (\colim_i \iota X_i) \cong \iota (\colim_i L \iota X_i) \cong \iota (\colim_i X_i).
  \end{align*}
  to the respective other isomorphism.

  Of course the same holds for coends,
  so applied to \eqref{eq:psh-sub-lan}, every coend $\Lan_\y(\iota F)(\iota C) \cong \int^{B} \iota(C)B \cdot \iota (F B)$  lies in $\A$ (is a fixed point of $\iota L$), if and only if for all $A$ \eqref{eq:psh-sub-lan} and \eqref{eq:psh-sub-nerve} are isomorphic, which by the argument above holds if and only if $\y'$ is well-behaved for $F$.
\end{proof}

We introduce criterions such that the extension of an action yields a right-closed monoidal structure.
The action has to behave similar to substitutions in the sense that $V \tl A$ models a $V$-family in $A$.

\begin{thm}[Monoidality of $\tl_J$]\label{thm:gen-mon}
  \fbnote{we previously had $J$ as a parameter with assumption $A \tl J B \cong J(A \otimes B)$, but this then gives $JA \cong J(A \otimes I) \cong A \tl JI$, so this can be simplified.
  What are conditions for when for an object $J_I$ the functor $JA = A \tl J_I$ is well-behaved as below?}
  Let $\A$ be a cocomplete category and let $\V$ be a small monoidal category $(\V, \otimes, I)$ acting on $\A$ via $\tl \colon \V \times \A \rightarrow \A$.
  Let $J_I \in \A$ be an object such that the functor
  \(J \colon \V \rightarrow \A\) defined by \(JV = V \tl J_I \) is well-behaved for all $(-) \tl A$, $A \in \A$.
  Then every natural isomorphism
  \begin{equation}
    \label{eq:J-ass}
    \aiso \colon (V \tl A) \tl_J B \cong V \tl (A \tl_J B),
  \end{equation}
  induces natural isomorphisms
  \begin{align}
    \label{eq:gen-mon-ass}
    \alpha \colon (A \tl_J B) \tl_J C &\cong A \tl_J (B \tl_J C) \\
    \lambda \colon JI \tl_J A \cong A, \quad&\quad \rho \colon A \tl_J JI \cong A \label{eq:gen-mon-units} \\
    \gamma \colon JV \tl_J A &\cong V \tl A. \label{eq:gen-mon-comp}
  \end{align}
\end{thm}

\begin{rem}\label{rem:coherence}
  \begin{enumerate}
    \item In \Cref{thm:gen-mon} the isomorphisms $a_{A, C, D}$ have coends as domains, so by choosing the coends appropriately we may without loss of generality assume that the $a$'s are in fact {equalities}.
          If we do the same for the its extension $\alpha$ and assume that the action $\tl$ is strict, then $(\A, \tl_J, JI)$ is a strict monoidal category.
          For an arbitrary natural isomorphism $a$ however, additional compatibility conditions on $\aiso$ are required (see the full version), so that $(\A, \tl_J, JI, \alpha, \lambda, \rho)$ satisfies the monoidal axioms.
    \item Proch\'azkov\'a~\cite{p25} recently showed how to construct skew-monoidal structures from actions of a skew-monoidal category under conditions similar to~\Cref{thm:gen-mon}.
          However, \cite{p25} requires an adjoint to $J$, so to apply it one would have to lift the action from $\V$ to the presheaf category $[\V^\op, \Set]$, which does not seem easier.
  \end{enumerate}
\end{rem}

\begin{thm}[Closedness of $\tl_J$]\label{thm:gen-mon-clo}
  In the setting of \Cref{thm:gen-mon}, suppose that $\A((-) \tl A, B)$ is a fixed point of the monad $\nu_J \Lan_J$, that is, we have an isomorphism \eqref{eq:asm-ra}
   \begin{equation}
     \label{eq:asm-ra}
     \nu_J \Lan_J \A((-) \tl A, B) \cong \A((-) \tl A, b)
   \end{equation}
   for all $A, B \in \A$.
   Then defining \((A \dwand B) = \Lan_J (\A((-) \tl A, B))\) yields an adjunction 
   \begin{equation}
     \label{eq:gen-mon-adj}
     \A(A \tl_J B, C) \cong \A(A, B \dwand C)
   \end{equation}
\end{thm}
We apply these constructions to $\PSh\F$ and related presheaf categories in the next section.

\begin{expl}[$\tl_J$ recovers Day convolution]\label{expl:gen-mon-day}
  Given a monoidal category $(\V, \otimes, I)$ we can set $\A = \PSh \V$ and take the action \(V \tl X = \y V \otimes X\)
  to recover Day convolution $\tl_J \cong \otimes_{\mathrm{Day}}$.
  Note that the assumptions of \Cref{thm:gen-mon} are satisfied since for $J_I = \y I$ we get $J = \y$ which is well-behaved by \iref{expl:wb}{y}.
\end{expl}

\begin{rem}
  \begin{enumerate}
    \item Well-behaved functors were introduced in a different way~\cite{acu14} to construct a monoidal structure on a functor category $[\D, \E]$.
          Given a well-behaved parameter $K \colon \D \rightarrow \E$, they define the tensor $\circ_K$ on $[\D, \E]$ via left Kan extension along $K$:
          \begin{equation*}
            X \circ_K Y = \Lan_K X \cdot Y \cong \int^{D} XD \cdot \E(KD, Y-) \qquad \text{for } X, Y \in [\D, \E].
          \end{equation*}
          The constructions $\circ_K$ and $\tl_J$ are orthogonal. Their intersection is the substitution tensor on $\PSh\F$, which corresponds to choosing $K=V\colon \F\to \Set$. However, unlike our $\tl_J$, the construction $\circ_K$ does not produce substitution tensors on other presheaf categories.
  \end{enumerate}
\end{rem}

\section{Substitution in Presheaves}
\label{sec:psh-substitution}
As a first family of applications of the general theory of \Cref{sec:mon-struct-from-act}, we present a uniform perspective on several substitution tensors in categories $\PSh{ \ctx }$ of presheaves over (untyped) variable contexts. We do so by focusing on the properties of index category $\ctx$. This approach is conceptually quite different from existing general accounts of substitution~\cite{p03,pt08}, where $\ctx$ is a freely generated category of contexts according to some  ``rule format'' $\mathcal{S}$ describing context manipulation.
Abstractly, this is modeled as $\ctx = \mathcal{S} 1$ for a pseudo-monad $\mathcal{S}$ on the category $\cat{Cat}$ of small categories.

\begin{nota}[Categories of contexts]\label{nota:subcat-fin-sets}
We denote by \(\I,\, \S,\, \B\hookrightarrow \F\) the non-full wide subcategories of $\F$ with morphisms given by injections, surjections, and bijections, resp.
\end{nota}

As hinted in \Cref{sec:substitution-intro}, objects of the respective categories are untyped contexts.
Morphisms represent the permitted context manipulations: exchange (\B), weakening + exchange~(\I), contraction + exchange (\S), weakening + exchange + contraction (\F).

\subparagraph*{Pullbacks.} Given a set $Y$ we denote by $\Set/ Y$ the slice category over $Y$, whose objects are maps $b \colon B \rightarrow Y$, and with morphisms $b \rightarrow c$ those maps $f \colon B \rightarrow C$ with $c \cdot f = b$.
Abusing notation, we sometimes identify a map $b \colon B \rightarrow Y$ with its family $(B_y)_{y \in Y}$ of fibres $B_y = b^{-1}(y)$.
The \emph{base change} of $b$ along $f$ is the projection \( f^* b \colon f^{*} {B} \rightarrow X\) of the pullback \eqref{eq:pullback} of \(b\) along \(f\).\vspace{-8pt}\\
\begin{minipage}{0.4\linewidth}
  \begin{equation}\label{eq:pullback}
    \begin{tikzcd}
      f^{*}B \rar{r_{f}} \pullbackangle{-45}  \ar{d}[swap]{f^* b} & B \dar{b} \\
      X \rar{f} & Y.
    \end{tikzcd}
  \end{equation}
\end{minipage}
  \hfill%
\begin{minipage}{0.4\linewidth}
  \begin{equation}
    \label{eq:dep-prod}
    p_f \colon \prod_y B_y \rightarrow \prod_x B_{f(x)}
  \end{equation}
\end{minipage}\\
It has fibers \((f^{*}B)_{x} = B_{f(x)}\), and we denote the \emph{re-indexing} projection \(f^{*}B \rightarrow B\) by \(r_{f}\).
Base change extends to a functor \(f^* \colon \Set / Y \rightarrow \Set / X\).
The functor $f^*$ has left and right adjoints given by \emph{dependent sum} and \emph{dependent product}:
the left adjoint $\sum_f \colon \Set / X \rightarrow \Set / Y$ sends $a \colon A \rightarrow X$ to $f \cdot a$, and the right adjoint $\prod_f$ is defined on fibers via $(\prod_f A)_y = \prod_{f(x) = y}A_x$.
If $f \colon X \rightarrow 1$ is the terminal map we write $\prod_f A = \prod_x A_x$. We denote by $p_f$ the map \eqref{eq:dep-prod} with $\pi_x \cdot p_f = \pi_{f(x)}$, where $\pi_{(-)}$ are product projections.

The following definition isolates the technical conditions on a category of contexts to admit a closed substitution structure:
\begin{defn}[Contextuality]\label{defn:contextual}
  A subcategory $\ctx \hookrightarrow \F$ is \emph{contextual} if (1) $\ctx$ contains all objects of $\F$, (2) $\ctx$ is closed under the monoidal structure $+$ of $\F$, and (3) $\ctx$ is pullback-stable: for every pullback square \eqref{eq:pullback} in $\F$, if $f\in \ctx$ then $r_f\in \ctx$.
\end{defn}
We emphasize that condition (2) only requires that $f, g \in \ctx$ implies $f + g \in \ctx$, \emph{not} that $+$ is a coproduct in $\ctx$.

\begin{expl}
  The categories $\F$, $\I$, $\S$, and \(\B\) are contextual.
  A more obscure example is given by the contextual subcategory of $\S$ generated by the map $3 \rightarrow 1$.
\end{expl}

An equivalent characterization of contextuality is given by:

\begin{prop}\label{prop:reindexing}
  Let \(\ctx \hookrightarrow \F\) be a wide subcategory of $\F$ containing all bijections that is closed under \(+\).
  Then $\ctx$ is contextual if and only if $\ctx$ is closed under \(\times\) and is \emph{prime}:
  \begin{center}
    for all $f, g \in \F$: \qquad $f+ g \in \ctx \Longrightarrow f, g \in \ctx$.
  \end{center}
\end{prop}

\begin{asm}\label{asm:subcat}
 In the remainder of \Cref{sec:psh-substitution}, we fix a contextual subcategory $\ctx \hookrightarrow \F$.
The category $\ctx$ is closed under the monoidal structures $+$ (by definition) and $\times$ (by \Cref{prop:reindexing}). We
denote by \(\oplus\), $\otimes$ the Day convolution on \(\PSh{\ctx}\) w.r.t.\ \(+\), $\times$.
\end{asm}

Consider a substitution $\sigma$ of variables $a$ from a finite set~$A$ by terms $t_a$.
The substitution $\sigma$ then depends on the variables \emph{contravariantly}, and on the variables contained in the terms $t_a$ \emph{covariantly}.
This is captured abstractly by the following definition, which generalizes the left action $\tl$ from $\PSh\F$ to $\PSh{ \ctx }$:
\begin{defn}[Substitution presheaf]\label{def:sub}
  For \(A \in \ctx\) and \(X \in \PSh{ \ctx }\) we define the \emph{substitution presheaf}
  \begin{equation}
    \label{eq:sub}
    A \tl X = \int^{(B \rightarrow A) \in \F / A} \y B \cdot \prod_{a}XB_{a}.
  \end{equation}
\end{defn}
Concretely, elements of $(A \tl X)C$ are equivalence classes of pairs \(\sigma = [f, (t_a)_{a \in A}]\) of some $f \in \ctx(B, C)$ and a family of elements $t_a \in XB_a$ subject to the equations $[f' \cdot j, (t_a)] = [f', (j_a \cdot t_a)]$ for every bundle map $j \colon B \rightarrow B'$ and every $f' \in \ctx(B', C)$.
The following presentation of $\tl$ was used implicitly on $\PSh \F$ by Fiore et al.~\cite{ftp99} and explicitly on $\PSh \B$ by Tanaka~\cite{t00}:
\begin{cor}\label{cor:conv-rep-sub}
  The substitution presheaf satisfies $A \tl X \cong \bigoplus_{a \in A} X$.
\end{cor}
For $\PSh{ \F }$ we have $\oplus = \times$, so we recover the action $A\tl X=X^A$ considered in \Cref{sec:substitution-intro}.
The contravariant functoriality of $A \tl X$ in $A$ now corresponds to reindexing of a substitution.
\begin{prop}\label{thm:day-pow-bifunct}
  Substitution \((-) \tl (-) \colon \ctx^{\op} \times \PSh{ \ctx } \rightarrow \PSh{ \ctx }\) is bifunctorial.
\end{prop}

We record some auxiliary properties of substitutions.

\begin{lem}\label{lem:subst-props}
  Let $A \in \ctx$ be a finite set and $X \in \PSh{ \ctx }$ be a presheaf.\vspace{6pt}\\
  \begin{minipage}{0.5\linewidth}
    \begin{enumerate}
      \item\label{lem:subst-props:right-1} \(A \tl \y 1 \cong \y A\).
      \item\label{lem:subst-props:0}   \(\emptyset \tl X \cong \y 0\)
      \item\label{lem:subst-props:sum} \((A + A') \tl X \cong (A \tl X) \oplus (A' \tl X)\).
    \end{enumerate}
  \end{minipage}
  \begin{minipage}{0.5\linewidth}
    \begin{enumerate}
      \item[(4)]\label{lem:subst-props:1} \(1 \tl X \cong X\).
      \item[(5)]\label{lem:subst-props:prod}   \(A \tl (B \tl X) \cong (A \times B) \tl X\)
      \item[(6)]\label{lem:subst-props:rep} \(A \tl \y B \cong \y (A \times B)\).
    \end{enumerate}
  \end{minipage}
\end{lem}

\begin{prop}\label{prop:day-power-lan}
  Let \(\iota \colon \ctx \hookrightarrow \ctx'\) be an inclusion of contextual subcategories.
  Then we have a natural isomorphism
  \(\Lan_{\iota} (A \tl X) \cong \iota A \tl (\Lan_{\iota} X)\).
\end{prop}

We obtain the substitution tensor on $\PSh{ \ctx }$ by instantiating \Cref{constr:ext} to the functor $J = \y\colon \ctx^\op\to \PSh{ \ctx }$:
\begin{defn}[Substitution tensor]
  The \emph{substitution tensor} on $\PSh{ \ctx }$ is given by the $\y$-extension of the action $\tl$:
  \begin{equation}\label{eq:pre-sub}
    X \sub Y := X\tl_\y Y = \int^{A \in \ctx} XA \cdot (A \tl Y).
  \end{equation}
\end{defn}
Intuitively, an element of $X \sub Y$ is a formal substitution of an $X$-term $x$ by $Y$-terms $y_{a}$, one for each variable $a$ in the context of $x$.
Explicitly, the elements of the set $(X \sub Y)(C)$ are equivalence classes $[x \in XA, \gamma \in (A \tl Y)C]$ of the equivalence relation generated by \[(j \cdot x, \gamma) \sim (x, \gamma \cdot j) \quad \text{for}\quad x \in XA,\, \gamma \in (A' \tl Y)C,\, j \in \ctx(A, A').\]

The following proposition is central:
\begin{prop}\label{prop:sub-pres-day}
  The functor \((-) \sub Y\) left-distributes over $\oplus$:
  \[(X \oplus Y) \sub Z \cong (X \sub Z) \oplus (Y \sub Z).\]
\end{prop}

We are ready to establish our main result on the substitution tensor on presheaves:
\Cref{thm:gen-mon} yields a closed monoidal structure.
The theorem below uniformly recovers the original instance of Johnstone and Wraith~\cite{jw78} and Fiore et.~al.~\cite{ftp99} for $\PSh \F$ and  the presentation by Kelly~\cite{k05o} (see also Tanaka~\cite{t00}) for $\PSh \B$.
In fact, for index categories such as $\B, \I, \S, \F$ that are equal to $M1$ for some pseudomonad $M$ on $\cat{Cat}$,  there exist multiple extensive accounts (e.g.~\cite{pt08,fghw08,c12}) of how to derive the substitution tensor on $[M1, \Set]$.
\begin{thm}\label{thm:subst-monoidal}
  Substitution yields a right-closed  monoidal category  $(\PSh{ \ctx },\sub,V)$  with unit $V=\ctx(1,-)$ and internal homs
  \(Y \dwand Z = \Nat((-) \tl Y, Z)\).
\end{thm}

\begin{rem}
  If one works with the skeleton $\mathbb{O} \subseteq \F $ of finite ordinals, then the requirement on $\ctx$ to contain all bijections can be dropped, and one still obtains a substitution tensor:
  Tronin~\cite{t02} studied \emph{verbal} subcategories of $\mathbb{O}$, which are closed only under $+$ and pullbacks along \emph{partitions} (viz.\ non-decreasing maps in $\mathbb{O}$). 
\end{rem}

Substitution is respected by extending the category of contexts:

\begin{prop}\label{thm:sub-pres}
  Let \(\iota \colon \ctx \rightarrow \ctx'\) be an inclusion of contextual subcategories. Then \(\Lan_{\iota} \colon \PSh{ \ctx }\to \PSh{ \ctx' }\) is strong monoidal for their substitution tensors.
\end{prop}

\subsubsection*{Uniform Substitution via Day Convolution}
We conclude this section with a connection between the substitution tensor to another, restricted kind of substitution on $\PSh{ \ctx }$.
The perspective of substitution as induced by an action of $(\ctx, \times, 1)^\op$ on $\PSh{ \ctx }$ leads us to interpret Day convolution $\otimes$ w.r.t.~$\times$ as \emph{uniform substitution} of terms.

For motivation, recall that a finitary monad on $\Set$ is the same as a monoid in the monoidal category $([\Set, \Set]_{\mathrm{fin}}, \circ, \Id)$ of finitary $\Set$-endofunctors with composition as tensor.
The latter category is, under $\Lan_\iota \dashv \iota^*$, monoidally equivalent to the category $(\PSh \F, \sub, \y 1)$.
Consequently, we have:
\begin{notheorembrackets}
  \begin{prop}[\cite{jw78,ftp99}]\label{prop:finitary}
    The category of finitary monads on $\Set$ is equivalent to the category of monoids in $(\PSh \F, \sub, \y 1)$.
  \end{prop}
\end{notheorembrackets}
For a contextual subcategory $\ctx \hookrightarrow \F$, monoids in the category $(\PSh{ \ctx }, \sub, \y 1)$ thus correspond to a subclass of finitary monads. For example, for $\ctx = \B, \S$ monoids in $\PSh \ctx$  correspond precisely to \emph{analytic monads}~\cite{j81} and \emph{preimage preserving} or \emph{regular}~\cite{sz15}  monads on $\Set$, respectively.
A related characterization is the following (see \cite[Sec.~6.3]{l21} for the case $\ctx = \B$):
\begin{prop}\label{prop:sub-monoids-char}
  The category of monoids in $(\PSh{ \ctx }, \sub, \y 1)$ is equivalent to the category of monads on \(\PSh{ \ctx }\) preserving colimits and Day convolution $\oplus$.
\end{prop}

We integrate the Day convolution $\otimes$ with respect to $\times$ on $\ctx$ into this setting:
\begin{defn}[Uniform substitution]\label{def:uniform-sub-presheaves}
  The \emph{uniform substitution tensor} $\otimes$ on $\PSh{ \ctx }$ is given by Day convolution with respect to~$(\ctx, \times, 1)$.
\end{defn}
To understand how $\otimes$ captures uniform substitution, consider two operations $s,t$ from a signature $\Sigma$ with respective arities $n$ and $m$.
Then we can substitute~$t$ at every position of $s$ to obtain the composite $s[t]$ of arity $nm$, which, applied to an $n \times m$-matrix $M_{ij}$ of variables, produces the term \[s[t](M) = s(t(M_1), \ldots, t(M_n)), \text{ where } t(M_i) = t(M_{i1}, \ldots, t(M_{in}))\]
inheriting the symmetries of both $s$ and $t$.

\begin{prop}\label{prop:unif-sub-embed}
  There is a morphism $\varphi \colon X \otimes Y \rightarrow X \sub Y$ of monoidal structures.
\end{prop}


This relation between $\otimes$ and $\sub$ allows us to characterize commutative \emph{finitary} monads as commutative monoids for the substitution tensor:
By precomposition with the morphism $\varphi \colon \otimes \rightarrow \sub$ from \Cref{prop:unif-sub-embed}, every $\sub$-monoid induces a $\otimes$-monoid.
This leads to the following characterization of commutative finitary monads in terms of (uniform) substitution:

\begin{notheorembrackets}
  \begin{thm}[{\cite[Corollary 40]{gf16}}]\label{thm:com-mon}
    Let $\mathcal{T}$ be a finitary monad on $\Set$, and let $M$ be the corresponding monoid in $(\PSh \F, \sub, \y 1)$ (\Cref{prop:finitary}). Then $\mathcal{T}$ is commutative iff the $\otimes$-monoid induced by $M$ is commutative.
  \end{thm}
\end{notheorembrackets}

\begin{rem}
  Garner and Franco~\cite{gf16} developed a general framework for commutativity based on \emph{duoidal categories}, from which they derived \Cref{thm:com-mon}.
  This suggests a deeper connection via duoidal categories between the monoidal structures as induced by actions and Day convolution.
\end{rem}

\section{Substitution in Nominal Sets}
\label{sec:subst-nomin-sets}
Guided by intuition about the substitution tensor on presheaves,
we next apply the method of \Cref{sec:mon-struct-from-act} to introduce the substitution tensors for nominal sets and renaming sets.

\subsection{Nominal Sets}
\label{sec:little-dict-betw-1}

We recall the basics of the theory of nominal sets; see Pitts' book~\cite{p13} for an introduction.
Fix a countably infinite set \names of \emph{names} or \emph{atoms}, and denote by $\Perm \names$ the group of finite permutations on $\names$, i.e.\ bijections $\pi \colon \names \to \names$ fixing all but finitely many names. A \emph{$\Perm \names$-set} is a set $X$ equipped with a group action $\Perm \names \times X\to X$, denoted $(\pi,x)\mapsto \pi\cdot x$. The \emph{orbit} of an element $x\in X$ is the set $\{\pi\cdot x \mid \pi\in\Perm\names \}$. A map $f \colon X \rightarrow Y$ between $\Perm \names$-sets is \emph{equivariant} if $\pi \cdot f(x) = f(\pi \cdot x)$ for all $x \in X$ and $\pi \in \Perm \names$.
$\Perm \names$-sets and equivariant maps form a category.
A $\Perm \names$-set is \emph{nominal} if every element $x \in X$ is \emph{finitely supported}, that is, there exists a finite subset $S \subseteq_\f \names$ that \emph{supports $x$}:
whenever $\pi \in \Perm \names$ fixes $S$, viz.\ $\pi_S = \id$, then $\pi \cdot x = x$.
If $S, T \subseteq_\f \names$ support $x$, then so does their intersection $S \cap T$; in particular, every element $x$ has a \emph{least} finite support, denoted $\supp x$.

The idea is that an element $x$ of a nominal set $X$ is some syntactic object containing free names, and $\supp x$ is the set of free names in~$x$. For example, the set of $\lambda$-terms modulo $\alpha$-equivalence with free variables from $\names$ forms a nominal set whose action permutes free variables (e.g.\ $(a\, b)\cdot \lambda a.\, a\, b = \lambda c.\, c\, a$). The least support of a $\lambda$-term is its set of free variables (e.g.\ $\supp (\lambda a.\, a\, b) = \{b\}$).
 
We let $\Nom$ denote the full subcategory of $\Perm \names$-sets given by nominal sets. It has colimits and limits, with colimits and finite limits constructed like in $\Set$, and it is cartesian closed (in fact a topos).
Another important closed monoidal structure on nominal sets is the \emph{fresh product} $X * Y \subseteq X \times Y$ containing only pairs with disjoint support.
If $(x, y) \in X * Y$ we write $x\#y$  and say that $x$ is \emph{fresh} for $y$.
The fresh product induces the functor \eqref{eq:y-nom}
\vspace{-4pt}\\
\begin{minipage}{0.45\linewidth}
\vspace{-7pt}
  \begin{equation}
    \label{eq:y-nom}
    \psi \colon \I^\op \rightarrow \Nom, \;\;\psi(A) =  \names^{*A}
  \end{equation}
\end{minipage}
\hfill%
\begin{minipage}{0.5\linewidth}
  \begin{equation}
    \label{eq:nom-adjs}
    \begin{tikzcd}
      \PSh \I
      \ar[yshift=-2pt, bend left=20]{r}{\sh}
      \ar[yshift=2pt, bend right=20, leftarrow, swap]{r}{\unsh}
      \ar[phantom]{r}{\rotatebox{-90}{\(\dashv\)}}
      & \Nom
      \ar[r, "\iota" {xshift=4pt}, yshift=-2pt, bend left=20]
      \ar[yshift=2pt, bend right=20, leftarrow, swap]{r}{(-)_\fs}
      \ar[r, phantom, "\rotatebox{-90}{\(\dashv\)}" {xshift=3pt}]
      & {\Perm \names\text{-set}}
    \end{tikzcd}
  \end{equation}
\end{minipage}\\
sending $A$ to the nominal set $\names^{*A} = \names * \cdots * \names$ with $|A|$ factors.

Nominal sets are encased between $\Perm \names$-sets and presheaves:
\begin{nota}
  For convenience, we identify the category $\F$ of finite sets with its equivalent full subcategory of finite subsets of $\names$.
  If $A \subseteq C$ we also write $(A \subseteq C)$ for the inclusion map.
\end{nota}

Nominal sets are related to both $\Perm \names$-sets and presheaves over $\I$ via the two adjunctions from \eqref{eq:nom-adjs}.
The right-hand adjunction expresses that nominal sets form a full coreflective subcategory of $\Perm \names$-sets, with the coreflector $(-)_\fs$ simply picking out the finitely supported elements of a $\Perm \names$-set.

The adjunction between $\Nom$ and $\PSh \I$~\cite[Chapter 6]{p13} is central for the understanding of substitution on nominal sets.
It is somewhat similar to the (discrete) \emph{Grothendieck construction} between fibrations and indexed categories:
A nominal set $X\in \Nom$ acts as the (fiber-discrete) ``total category'' over $\I$, with fibration $\supp \colon X \rightarrow \I$. The corresponding (discrete) ``indexed category'' is the presheaf $I_* X \colon \I \rightarrow \Set$. More precisely:

\begin{constr}[Adjunction $I^*\! \dashv I_*$]\label{constr:nom-psh}~
  \begin{enumerate}
    \item \(\unsh \colon \Nom \rightarrow \PSh \I\) takes a nominal set \(X\) to
          \(\unsh X \colon A \mapsto \{x \in X \mid A \text{ supports } x\}\).
    \item \(\sh \colon \PSh \I \rightarrow \Nom\) takes a presheaf \(F\) to 
          \( \{(A, x) \mid A \subseteq_{\f} \names,\, x \in FA\} /\!\! \sim
          \),
          where the equivalence relation $\sim$ is given by \((A, x) \sim (A', x')\) if there is a finite set \(C\) with \(A, A' \subseteq C\) and $(A\subseteq C) \cdot x = (A'\subseteq C) \cdot x'$ in \(FC\).
          We denote the equivalence class of \((A, x)\) by \([A, x]\).
  \end{enumerate}
\end{constr}

\begin{notheorembrackets}
  \begin{thm}[{\cite[Thm.~6.8]{p13}}]\label{thm:nom-sh}
    The functor \(\unsh\) is a fully faithful right adjoint. Its left adjoint \(\sh\) preserves finite limits.
    The image of $\unsh$ is the full subcategory $\PShpb \I \subseteq \PSh \I$ of intersection-preserving presheaves.
  \end{thm}
\end{notheorembrackets}
Intersection preservation captures precisely the corresponding property of finite supports.
\begin{nota}
  Let $X$ be a nominal set.
  By \Cref{constr:nom-psh}, every \(x \in X\) is an element of every set \((\unsh X)A\) with \(\supp x \subseteq A\).
  We abuse notation to write $j \cdot x = \pi_j \cdot x$ for $j \colon A \rightarrow B$ and $x \in X$, where $\pi_j \in \Perm \names$ is any permutation agreeing with $j$ on $\supp x$.
  Conversely, for \(F \in \PSh \I\) with \(x \in FA\) and \(\pi \in \Perm \names\) we write \[\pi \cdot x = \pi|^B_{A} \cdot x \in FB \]
whenever the bijection $\pi$ (co-)restricts to an injection $\pi|^B_A\colon A \mono B$.
\end{nota}

\begin{lem}\label{lem:nom-nerve}
\Cref{constr:nom-psh} is precisely the nerve \eqref{eq:nerve} of the functor $\psi\colon \I^\op\to \Nom$ defined in \eqref{eq:y-nom}:
  We have isomorphisms $\unsh \cong \nu_\psi$ and $\sh = \Lan_\y \psi$.
\end{lem}

\subsection{Substitution Tensor in Nominal Sets}

In this section we construct the substitution tensor on nominal sets.
Inspired by the corresponding tensor on presheaves we derive it via~\Cref{thm:gen-mon}, and show that it has an explicit simple description due to the finite support property of nominal sets.
We start by defining the left action on $\Nom$.
Since \Nom is a subcategory of $\PSh \I$, the candidate is clear, as the action $\tl\colon \I^\op \times \PSh \I \rightarrow \PSh \I$ over presheaves is given by iterated Day convolution $\oplus$ (\Cref{cor:conv-rep-sub}).
It has been noted informally by several authors~\cite{m20, c13} that Day convolution $\oplus$ on $\PSh \I$ is closely related to the fresh product $*$ on nominal sets.
The precise connection is as follows:

\begin{prop}\label{prop:day-fresh}
  Both the left $\sh \colon (\PSh \I, \oplus) \rightarrow (\Nom, *)$ and the right adjoint functor $\unsh \colon (\Nom, *) \rightarrow (\PSh \I, \oplus)$ are strong monoidal.
\end{prop}
We therefore define the left action $\tl$ via iterated fresh product:
\begin{equation}
  \label{eq:nom-act}
 \tl \colon \I^\op \times \Nom \rightarrow \Nom, \quad A \tl X = X^{* A} = \{f \in \Set(A, X) \mid \forall (a \ne b) \colon f(a) \# f(b)\}.
\end{equation}
\Cref{constr:ext} applied to the action $\tl$ and the functor $J=\psi\colon \I^\op\to \Nom$ of \eqref{eq:y-nom} yields:
\begin{defn}[Substitution tensor]\label{defn:nom-sub}
  The \emph{substitution tensor} on $\Nom$ is given by the \(\psi\)-extension of $\tl$:
  \begin{equation}
    \label{eq:nom-sub-abstr}
    X\sub Y := X \tl_\psi Y = \int^{A \in \I} \nu_\psi X A \cdot Y^{*A}.
  \end{equation}
\end{defn}

We give an elementary description of \eqref{eq:nom-sub-abstr}, which also makes it easier to verify the conditions of \Cref{thm:gen-mon}.
Generalizing \eqref{eq:nom-act} to nominal sets $Y$ and arbitrary sets $A$ we write $Y^{*A}$ for the $\Perm \names$-set
\begin{equation}
  \label{eq:fresh-power}
  Y^{* A} = \{f \colon A \rightarrow Y \mid a \ne b \text{ implies } f(a) \# f(b)\}
\end{equation}
with the pointwise action.
Note that $Y^{*A}$ is nominal if $A$ is finite.

\begin{thm}\label{thm:nom-sub-pres}
  The substitution tensor for $X,Y\in \Nom$ is given by
  \begin{equation}\label{eq:sub-pres} X \sub Y = \{(x, \gamma) \mid x \in X,\, \gamma \in Y^{* \supp x}\}/\!\! \sim, \qquad \text{ where } (\pi \cdot x, \gamma) \sim (x, \gamma \pi),
  \end{equation}
  for \(x \in X,\, \gamma \in Y^{*\supp(\pi \cdot x)}\).
  We write \(x[\gamma]\) for the equivalence class of \((x, \gamma)\).
  The $\Perm \names$-action and least supports of elements of $X \sub Y$ are given by
  \[\pi \cdot x[\gamma] = x[\pi \cdot \gamma], \qquad \supp x[\gamma] = \supp \gamma = \bigcup_{a \in \supp x} \supp \gamma(a).\]
\end{thm}

\begin{thm}\label{thm:nom-sub}
 $(\Nom,\sub,\names)$ is a monoidal category with unit $\psi 1\cong \names$.
\end{thm}
\begin{proof}[Proof sketch]
  We apply \Cref{thm:gen-mon} to the setting where
  $\names$ is the dual of the monoidal category $(\I, \times, 1)$;
    $\A$ is the category $\Nom$;
    $J=\psi\colon \I^\op\to \Nom$ is given by \eqref{eq:y-nom};
    and the action $\tl\colon \I^\op\times \Nom\to \Nom$ is given by \eqref{eq:nom-act}.
\end{proof}
\begin{rem}
  Since $I_* (\Nom, *) \rightarrow (\PSh \I, \oplus)$ is monoidal by \Cref{prop:day-fresh}, the embedding $\unsh \colon \Nom \incl \PSh \I$ makes $\Nom$ a monoidal subcategory of $\PSh \I$.
  This yields an alternative way to calculate the substitution of nominal sets $X, Y$, namely by calculating it in presheaves: \[X \sub Y =\sh (\unsh X \sub \unsh Y).\]
\end{rem}
It remains to show that the substitution tensor gives a \emph{closed} monoidal structure. By \Cref{thm:gen-mon-clo}, this requires for every presheaf of the form $\Nom(X^{*(-)}, Y) $  an isomorphism
\[(\nu_\psi \cdot \Lan_\y \psi)(\Nom(X^{*(-)}, Y)) \cong \Nom(X^{*(-)},Y).\]
By \Cref{lem:nom-nerve} and \Cref{thm:nom-sh}, this is equivalent to the presheaf $\Nom(X^{*(-)}, Y)$ preserving intersections.
We provide an explicit description of the nominal set corresponding to that presheaf. It is based on the following concept:


\begin{defn}[Finitely reducible map]\label{def:nom-sub-clo}
  Let $X$ and $Y$ be $\Perm \names$-sets.
  A map $f \colon X^{*\names} \rightarrow Y$ is \emph{(finitely) reducible} if there exists a finite subset $A \subseteq \names$ such that $f$ factors as \[f = g \cdot p_A \colon X^{*\names} \rightarrow X^{*A} \rightarrow Y\] for some $g \colon X^{*A} \rightarrow Y$, where $p_A = X^{(A \subseteq \names)}$ is projection.
  Finite reducibility for maps $X^\names \rightarrow Y$ is defined analogously.
  We define
\begin{equation}\label{eq:dwand} X \dwand Y = \{f \in \Perm \names\text{-set}(X^{*\names}, Y) \mid f \text{ is finitely reducible}\}.
\end{equation}
\end{defn}

The intuition behind a map \(f \colon X^{*\names} \rightarrow Y\) being finitely reducible is that the definition of $f$ only uses finitely many input components of an input sequence $\gamma \in X^{*\names}$.
\begin{prop}\label{prop:dwand}
  Let $Y$ and $Z$ be nominal sets.
  \begin{enumerate}
    \item The set \(Y \dwand Z\) is a nominal set under the action \[(\pi \cdot f)(\gamma) = f(\gamma \pi) \quad \text{ for } \quad \pi \in \Perm \names \text{ and } f \in Y \dwand Z.\]
          It satisfies that $\supp f \subseteq A$ if and only if $f$ factors through $p_A$.
    \item\label{prop:dwand:iso} We have $I_*(Y \dwand Z) \cong \Nom(Y^{*(-)}, Z)$.
  \end{enumerate}
\end{prop}
From \iref{prop:dwand}{iso} can now conclude:
\begin{thm}\label{thm:nom-sub-closed} The monoidal category $(\Nom,\sub,\names)$ is right closed with internal hom given by $\dwand$, that is:
  \(\Nom(X \sub Y, Z) \cong \Nom(X, Y \dwand Z)\).
\end{thm}
  The isomorphisms send equivariant maps \(f \colon X \sub Y \rightarrow Z\) and
  \(g \colon X \rightarrow (Y \dwand Z)\) to the maps $\curry_{\sub}(f)$ and $\uncurry_{\sub}(g)$, respectively, defined by
  \[\curry_{\sub}(f)(x)(\gamma) = f(x[\gamma|_{\supp x}]) \qquad \text{ and }\qquad  \uncurry_{\sub}(g)(x[\gamma]) = g(x)(\hat{\gamma}),\]
  where \(\hat{\gamma} \in Y^{* \names}\) is any function with \(\hat{\gamma}|_{\supp x} = \gamma\).

\begin{rem}\label{rem:nom-sub-cap}
  There also exists a ``captureful'' substitution functor on \(\Nom\) given by
  \[X \,\hat{\sub}\, Y = \{(x \in X, \,\gamma \in Y^{\supp x})\}/ \sim \quad \text{ where } \quad (\pi \cdot x, y) \sim (x, \gamma \pi),\]
  and \((-) \,\hat{\sub}\, Y\) has a right adjoint defined analogous to $\dwand$ in \eqref{eq:dwand}, with $Y^\names$ in lieu of $Y^{*\names}$.
  However, this structure is not monoidal as there is no right unit.
  For example, $\names^{*2} \,\hat{\sub}\, \names \cong \names^2$.
\end{rem}

\subsubsection*{Uniform Substitution}
We have seen in \Cref{sec:psh-substitution} that the product monoidal structure $\times$ on $\I$ induces, via Day convolution, a monoidal structure $\otimes$
on presheaves that models \emph{uniform} substitution (\Cref{def:uniform-sub-presheaves}) and naturally relates to the substitution tensor (\Cref{prop:unif-sub-embed}).
An analogous uniform substitution structure can also be introduced to the world of nominal sets:

\begin{defn}[Uniform substitution]
  The \emph{uniform substitution tensor} of nominal sets \(X\) and \(Y\) is the nominal set given by
  \[X \otimes Y = \{x[\gamma] \in X \sub Y \mid \gamma[\supp x] \text{ is contained in a single orbit}\}.\]
\end{defn}
While this definition of $\otimes$ is useful to see the connection to the substitution tensor,
 it hides the symmetry of $\otimes$. To expose it, we relate the nominal uniform substitution tensor $\otimes$ to its counterpart on the presheaf category $\PSh \I$.
 
For \(x[\gamma] \in X \otimes Y\), we can choose a fresh $y$ in the orbit containing all elements of $\gamma[\supp x]$. By abuse of notation we write $x[y]$ for $x[\gamma]$. Given $x\in X$ and $y\in Y$ with support sizes $|\supp x| = n, |\supp y| = m$,
the support of \(x[y] \in X \otimes Y\) may be pictured as a \(n\times m\)-matrix
\[
  \begin{pmatrix}
    j_{a_1}(b_1) & \cdots & j_{a_1}(b_m) \\
    \vdots & \ddots & \vdots \\
    j_{a_n}(b_1) & \cdots & j_{a_n}(b_m) 
  \end{pmatrix}
\]
with pairwise fresh names $j_{a_i}(b_j)$ as entries.
The element $x[y]$ then corresponds to this matrix modulo the internal symmetries of $x$ and $y$ allowing permutations of the rows and columns. 
We have essentially described the isomorphism of the following proposition:

\begin{prop}\label{prop:nom-univ-sub-from-presh}
  Uniform substitution is induced by Day convolution on presheaves:
  \[X \otimes Y \cong \sh(\unsh X \otimes \unsh Y) \cong \int^{A, B \in \I} \names^{*(A \times B)} \cdot \unsh X A \cdot \unsh Y B.\]
  In particular, $\otimes$ is symmetric.
\end{prop}

As a consequence, we get an analogue of \Cref{prop:day-fresh}:
\begin{prop}\label{prop:nom-unif-sub-mon}
  Both adjoints $\sh \colon (\PSh \I, \otimes) \rightarrow (\Nom, \otimes)$ and $\unsh \colon (\Nom, \otimes) \rightarrow (\PSh \I, \otimes)$ are strong monoidal.
\end{prop}

The internal hom of $\otimes$ restricts to the internal hom of $\sub$:
\begin{defn}
  Let \(Y^{*_{\orb} \names} \subseteq Y^{*\names}\) be the \(\Perm \names\)-set of all maps $f\in Y^{*\names}$ whose image $f[\names]$ is contained in a single orbit.
  \[Y \uswand Z = \{f \in \Perm \names\text{-set}(Y^{*_{\orb} \names},Z) \mid f \text{ is finitely reducible}\}\]
  with double precomposition as action: $(\pi \cdot f)(\gamma) = f(\gamma \cdot \pi)$.
\end{defn}

\begin{prop}
  $(\Nom,\otimes,\names, \uswand)$ is a symmetric monoidal closed.
\end{prop}


\subsection{Substitution Tensor in Renaming Sets}

Renaming sets were introduced~\cite{gh08} as an alternative to nominal sets which allows for renamings that are not necessarily injective.
Many of the basic definitions are similar to nominal sets:
A \emph{renaming} of \(\names\) is a map \(\rho \colon \names \rightarrow \names\) with \( \rho a = a\) for all but finitely many \(a \in \names\). The monoid of all renamings is denoted \(\rens \names\).
A \emph{(nominal) renaming set} consists of a set \(X\) together with a monoid action \(\rens \names \times X \rightarrow X\) such that every \(x \in X\) has a finite support, that is,
there exists a finite set \(A\) such that whenever $\rho \in \rens \names$ satisfies $\rho|_A = \id_A$ then $\rho \cdot x = x$.
This implies that $x$ has a \emph{least} finite support $\supp x$. The category of renaming sets and \(\rens \names\)-equivariant maps is denoted \(\Ren \).

There is an adjunction $\sh \dashv \unsh \colon \Ren \rightarrow \PSh \F$, defined analogous to \Cref{constr:nom-psh}, that co-restricts to an equivalence of categories:

\begin{notheorembrackets}
  \begin{thm}[{\cite[Sec.~5]{gh08}}]
    The category $\Ren$ is equivalent to the full subcategory \(\PShpb \F \subseteq \PSh \F\) of intersection-preserving presheaves.
  \end{thm}
\end{notheorembrackets}

The substitution tensor is even simpler for renaming sets than for nominal sets:
\begin{rem}
  Day convolution $\oplus$ in \(\PSh \F\) is just cartesian product since $+$ is a coproduct in $\F$.
  Thus, since the functor $\sh$ preserves finite limits, we have for all $X, Y \in \Ren$ that
  \[\sh(\unsh X \oplus \unsh Y) \cong \sh (\unsh X \times \unsh Y) \cong \sh \unsh X \times \sh \unsh Y \cong X \times Y\]
  is also the cartesian product.
  This answers the question~\cite{m20} why there is no fresh product on renaming sets.
  Moreover, it yields a conceptual argument for the result \cite[Thm~3.7]{m20} that the free-renaming-set construction \((\Nom, *) \rightarrow (\Ren, \times)\) is strong.
\end{rem}
We therefore choose the left action
\(\tl \colon \F^\op \times \Ren \rightarrow \Ren, (A \tl X) = X^A\).
 Extending it under the dense embedding
\(\psi \colon \F^\op \rightarrow \Ren, A \mapsto \names^A\),
satisfying $\nu_\psi \cong \unsh$,
yields the substitution tensor
for renaming sets:
Its definition is slightly more complicated than in $\Nom$ since renamings may not preserve supports, see~\Cref{sec:relev}.

\begin{defn}[Substitution tensor] The \emph{substitution tensor} on $\Ren$ is given by the \(\psi\)-extension of $\tl$:
  \begin{equation}X \sub Y := X\tl_\psi Y = \{(x \in X, \gamma \in Y^{\supp x})\}/ \sim,  \quad \text{where $(\sigma \cdot x, \gamma|_{\supp \sigma \cdot x}) \sim (x, \gamma \sigma)$} \end{equation}
  for $\sigma \in \rens \names$ and $\gamma \in Y^{\rho [\supp x]}$ generates the equivalence.
  We write \(x[\gamma]\) for the equivalence class of \((x, \gamma)\).
  Its action and support are given by
  \[\sigma \cdot x[\gamma] = x[\sigma \cdot \gamma] \quad \text{ and } \quad  \supp x[\gamma] = \bigcup_{a \in \supp x}\supp \gamma_{a}.\]
  The internal substitution hom for renaming sets is given by
  \begin{equation}Y \dwand Z = \{f \in (\rens \names\text{-set})(Y^{\names}, Z) \mid f \text{ is finitely reducible}\} \end{equation}
\end{defn}
An argument analogous to that for nominal sets shows that:
\begin{thm}
  $(\Ren,\sub,\names,\dwand)$ is a right closed monoidal category.
\end{thm}

\section{Nominal Sets and Presheaves: A Taxonomy}
\label{sec:taxon-nomin-sets}

In the previous section we used the correspondence between pre\-shea\-ves and nominal sets to discover new structural operations on nominal sets, such as substitution.
It is therefore worthwhile to study the connection between both worlds in more detail.
Our work builds on the note~\cite{f01} investigating presheaves over contextual categories, particularly \Cref{sec:nom-b}.
We first consider the nominal models corresponding to the two remaining presheaf categories $\PSh \B$ and $\PSh \S$ (representing abstract linear and relevant syntax), and will show that they both admit natural characterizations as subcategories of $\Nom$ or $\Ren$, respectively.
After considering alternative descriptions via sheaves and the related category of supported sets, we paint a complete picture relating the various arising variants of nominal sets. The results obtained in this section are summarized by the diagram in \Cref{sec:nominal-world}.

\subsection{Support Preservation and $\PSh \B$}
\label{sec:nom-b}

The presheaf category $\PSh \B$, also known as the category of \emph{(combinatorial) species}~\cite{j81}, serves as a natural model of \emph{linear} structures, e.g.~the linear $\lambda$-calculus~\cite{t00}. Here variables are interpreted as resources that are supposed to be used exactly once by terms in a given linear context. It turns out that on the side of nominal sets, linearity is precisely captured by the requirement that maps are not allowed to discard any names occurring in their input. More precisely, we consider the following non-full subcategory of $\Nom$:
\begin{defn}[$\Nom_=$]
  The category $\Nompres$ has nominal sets as objects, and morphisms are those equivariant maps $f\colon X\to Y$ that are \emph{support-preserving}, that is, $\supp f(x) = \supp x$ for all $x\in X$.
\end{defn}
Note that $\supp f(x) \subseteq \supp x$ holds for all equivariant maps. We consider the following modification of \Cref{constr:nom-psh}:

\begin{constr}[$\coprod$, $\mathcal{S}$]\label{constr:psh-nom-pres}
  \begin{enumerate}
    \item The functor $\coprod \colon \PSh \B \rightarrow \Nompres$ maps a presheaf $F \in \PSh \B$ to the nominal set $\coprod_{A \in \B}FA$ with action $\pi \cdot \kappa_A(x) = \kappa_{\pi[A]}(\pi|_A^{\pi[A]} \cdot x)$.
    \item Conversely, the functor $\mathcal{S} \colon \Nompres \rightarrow \PSh \B$ maps $X$ to the presheaf $\mathcal{S} X$ with
          \[\mathcal{S}X(A) = \{x \in X \mid \supp x = A\}, \quad  \mathcal{S}X(f)(x) = \pi_f \cdot x \text{ for } \pi_f \in \Perm \A \text{ extending } f \colon A \cong B.\]
          Note that defining  $\mathcal{S}$ on morphisms $f \colon X \rightarrow Y$ requires support-preservation, so that $f$ (co-)restricts for $A \in \B$ to a map \(\mathcal{S}(f)_A \colon \mathcal{S}X(A) \rightarrow \mathcal{S}X(B)\).
  \end{enumerate}
\end{constr}

\begin{lem}\label{lem:coprod-supp}
  For $F \in \PSh \B$ we have $\supp \kappa_A(x)= A$ for $x \in FA$.
\end{lem}

\begin{prop}\label{prop:spec-nompres-equiv}
  The functors $\coprod$ and $\mathcal{S}$ are an equivalence of categories: $\PSh \B \simeq \Nompres$.
\end{prop}
\begin{proof}
  The isomorphism \(\coprod \mathcal{S} X \cong X\) may be seen as the definition of a nominal set: every element has a (necessarily unique) least finite support. The isomorphism \(\mathcal{S}\coprod F \cong F\) holds by \Cref{lem:coprod-supp}:
  \begin{align*}
    \mathcal{S}(\coprod_{A \in \B} FA)(B) = \{\kappa_{A}(x) \mid x \in FA,\, \supp \kappa_{A}(x) = B\} \cong FB. \tag*{\qedhere}
  \end{align*}
\end{proof}

We thus obtain the following square of adjunctions:
\[
  \begin{tikzcd}[column sep=3em, row sep=3em]
    \Nompres
    \ar[yshift=0pt, bend left=20]{d}{\mathcal{S}}
    \ar[yshift=0pt, bend right=20, leftarrow, swap]{d}{\coprod}
    \ar[phantom]{d}{\rotatebox{-90}{\(\simeq\)}}
    \ar[yshift=-2pt, bend left=20, hook]{r}{i}
    \ar[yshift=2pt, bend right=20, leftarrow, swap]{r}{R}
    \ar[phantom]{r}{\rotatebox{-90}{\(\dashv\)}}
    & \Nom
    \ar[yshift=0pt, bend left=20]{d}{\unsh}
    \ar[yshift=0pt, bend right=20, leftarrow, swap]{d}{\sh}
    \ar[phantom]{d}{\dashv}
    \\ \PSh \B
    \ar[yshift=-2pt, bend left=20]{r}{\PSh{ \iota }}
    \ar[yshift=2pt, bend right=20, leftarrow, swap]{r}{\iota^*}
    \ar[phantom]{r}{\rotatebox{-90}{\(\dashv\)}}
    & \PSh \I
  \end{tikzcd}
\]
Here $\iota\colon \B\hookrightarrow \I$ is the inclusion, and the adjunction $i\dashv R$ on top is the composite of the three other adjunctions. Moreover:

\begin{prop}\label{prop:nom-psh-square}
  \begin{enumerate}
    \item We have $\PSh{ \iota }(F)B \cong \coprod_{A \subseteq B}FA$.
    \item The right adjoint $R = \coprod \cdot \iota^* \cdot \unsh$ is given by
          \[RX \cong \{(x, A) \mid x \in X,\, \supp x \subseteq A\}.\]
    \item The left adjoint $i = \sh \cdot \Lan_{ \iota } \cdot \mathcal{S}$ is isomorphic to inclusion.
    \item $\Nom$ is isomorphic to the Kleisli category of $\mathcal{T}$, that is,
  \( \Nom \cong (\Nom_=)_\mathcal{T}\).
  \end{enumerate}
\end{prop}

\Cref{prop:nom-psh-square} is just a nominal rephrasing of Fiore's~\cite{f01,fm05} result characterizing the topos $\PShpb \I$ as the Kleisli category for the monad $\iota^{*} \cdot \Lan_{ \iota }$  on $\PSh \B$.
The intuition is the right adjoint $R$ remembers the names dropped by an equivariant map, making it support-preserving.
Concretely, an equivariant map $f \in \Nom(X, Y)$ corresponds precisely to the support-preserving equivariant map $f_= \in \Nompres(X, RY)$ with $f_=(x) = (f(x), \supp x)$.
\Cref{prop:spec-nompres-equiv} shows that nominal sets with support-pre\-ser\-ving maps also yield a topos.
Let us provide some intuition on $\Nom_=$.



\subparagraph{Limits and Colimits.} Colimits in \Nompres are formed as in \Nom, but limits look very differently: Products in $\Nom_=$ are given by
\(X \times_{=} Y = \{(x, y) \in X \times Y \mid \supp x = \supp y\}\),
and the terminal object \(1_{=}\)   in \Nompres is the nominal set \(\powf \names\) with terminal map \(!_{=}\) given by  \( \supp \colon X \rightarrow \powf \names \).
\subparagraph{Fresh Product, Substitution and Abstraction.} Fresh product \(*\) in \(\Nompres\) is formed in $\Nom$; note that the unit of the fresh product in $\Nompres$ is \emph{not} the terminal object $1_{=}$ but the singleton nominal set $1$.


Substitution in $\Nompres$ is also formed in $\Nom$; in the presheaf category \(\PSh \B\) this was studied by Joyal~\cite{j81} as composition of species and interpreted as substitution by Tanaka~\cite{t00}.

Nominal sets permit a natural description of \emph{name abstraction} $[\names](-)$ modeling bound variables modulo $\alpha$-equivalence like in the $\lambda$-calculus~\cite{gp99,p13}.
For example, the nominal set $[\names]\names^2$ contains the elements $x = \langle a \rangle (a, b)$ and $y = \langle c \rangle (e, f)$, where the variables $a$ and $c$ are \emph{abstracted} (bound).
This means that in $[\names]\names^2$ these elements are \emph{equal} to $x = \langle d \rangle (d, b)$ and $y = \langle a \rangle (e, f)$, respectively.
Presheaves over $\B$ model \emph{linear variable binding}~\cite{t00,fr25}, meaning that all variables from the context are used exactly once.
Therefore the abstraction functor $\strabs (-)$ in $\Nompres \simeq \PSh \B$ is also \emph{linear}, viz.\ $a \in \supp x$ for every element $\langle a \rangle_= x \in \strabs X$.
For example, $\langle a \rangle_= (a, b) \in \strabs \names^2$, but $\langle c \rangle_= (e, f)$ does not make sense since $c \not \in \supp (e, f)$.

Finally, we give a simple description of the monad $\mathcal{T}$ on $\Nom_=$, using that the terminal object $1_=$ of $\Nompres$ is a monoid for the fresh product $*$ with (disjoint) union as multiplication.

\begin{prop}\label{prop:mon-desc}
  The monad \(\mathcal{T}\) is isomorphic to the writer monad of $*$ for the monoid $1_=$, that is, \(\mathcal{T} \cong 1_{=} * (-)\).
\end{prop}

\subsection{Relevance Sets and $\PSh \S$}
\label{sec:relev}
Next, we identify the nominal counterpart of the presheaf category $\PSh \S$. The latter is a suitable model for \emph{relevant} syntax, capturing notions of terms or computations where every variable (resource) in a given context is required to be used at least once~\cite{fr25}. Since $\S$ involves non-injective renamings, one would expect a close relation to renaming sets. However, for the latter a curious phenomenon can occur that prevents them from being equivalent to $\PSh \S$: unlike for permutations, applying renamings can drop part of the support, that is, it can happen that $\supp (\rho \cdot x)$ is a proper subset of  $\rho \cdot \supp x$.
\begin{expl}
  The restriction of the free-group monad to finite sets is an intersection-preserving presheaf $G \in \PSh \F$, and therefore corresponds to a renaming set $\sh G$ with
  \[\supp \rho \cdot ab^{-1} = \supp cc^{-1} = \supp 1 = \emptyset \ne \{c\} = \rho \cdot \supp ab^{-1}\]
  for the renaming $\rho = [a \mapsto c, \, b \mapsto c]$ and $ab^{-1} \in G\{a, b\}$.
\end{expl}

Gabbay and Hofmann~\cite{gh08} point out that they are not sure whether the existence of such models is a bug or a feature of the category $\Ren$. However, clearly such models are against the spirit of structures where all resources are expected to be relevant. Thus, for our purposes, it is a natural approach to exclude them:

\begin{defn}[Relevance set]
  A renaming set $X$ is a \emph{relevance set} if \(\rho \cdot \supp x = \supp (\rho \cdot x)\) for all $\rho\in \rens \names$ and $x\in X$.
  The corresponding full subcategory of \(\Ren \) is denoted by \(\Relev\), and its subcategory of support-preserving maps is denoted by $\Relevpres$.
\end{defn}

\begin{constr}
  We denote by $\coprod^{\S} \colon \PSh \S \rightarrow \Relevpres$ and $\mathcal{S}^\S \colon \Relevpres \rightarrow \PSh \S$ the functors defined as in \Cref{constr:psh-nom-pres}.
\end{constr}
Note that $\mathcal{S}^\S$ is well-defined by the relevance set condition.
While studying the corresponding presheaf categories, Fiore~\cite{f01} observed that the monad induced by the left adjoint $\Lan_\iota \colon \PSh \S \rightarrow \PSh \F$ is the lifting of the monad $\mathcal{T}$ on $\PSh \B \simeq \Nompres$.

\begin{thm}\label{thm:relev}
  \begin{enumerate}
    \item The functors $\coprod^{\S}$ and $\mathcal{S}^{\S}$ constitute an equivalence of categories: $\Relevpres \simeq \PSh \S$.
    \item The monad $\mathcal{T}$ on $\Nompres$ lifts to a monad $\overbar{\mathcal{T}}$ on $\Relevpres$ , and the category $\Relev$ is isomorphic to the Kleisli category of $\overbar{\mathcal{T}}$, that is,
          \(\Relev \cong (\Relevpres)_{\overbar{\mathcal{T}}}\).
  \end{enumerate}
\end{thm}

Relevance sets embed into presheaves:
Recall that \emph{preimages} are pullbacks where one leg is mono.
Let $\PShpi \F \subseteq \PShpb \F$ be the full subcategory preimage-preserving presheaves.

\begin{prop}\label{prop:preim-relev-equiv}
  The equivalence $\PShpb{\F}\! \simeq\! \Ren$ restricts to $\PShpi{\F}\! \simeq\! \Relev$.
\end{prop}

The equivalence to relevance sets allows us to recover the following characterization of the finitary preimage-preserving endofunctors on $\Set$:
\begin{notheorembrackets}
  \begin{thm}[{\cite{sz15}}]\label{thm:es-im-finit}
    The essential image of $\Lan_i \colon \PSh{\S} \rightarrow [\Set, \Set]$ for $i \colon \S \incl \Set$ are the
    finitary, preimage-preserving endofunctors on $\Set$.
  \end{thm}
\end{notheorembrackets}
The categorical structure of $\Relevpres$ is similar to that of $\Nompres$,
with again an important difference being the monoidal structure corresponding to Day convolution $\oplus$ for $+$.
Note that $\Relevpres$ has a monoidal structure $\oplus$ taking products of underlying sets $|X \oplus Y| \cong |X| \times |Y|$ with a singleton unit, which, like in $\Nompres$, is not the categorical product.
Moreover, the adjunction between nominal sets and renaming sets restricts to relevance sets.

\begin{prop}\label{prop:day-relev}
  Both the equivalence $(\Relevpres, \oplus) \simeq (\PSh \S, \oplus)$ and the inclusion $(\Relevpres, \oplus)\rightarrow (\Relev, \times)$ are strong monoidal.
\end{prop}

\begin{rem}\label{rem:sub-relev}
  We note that the substitution tensor $\sub$ of $\Ren$ restricts to relevance sets, that is, if $X$ and $Y$ are relevance sets, then so is $X \sub Y$.
  Note that this is not true for the internal hom: for $Y =  2$ the set $Y \dwand Y$ is not a relevance set.
\end{rem}

\subsection{Sheaves over Contexts}
\label{sec:sheaves}

The characterization of the essential image of $\unsh \colon \Nom \incl \PSh \I$ as intersection-preserving presheaves (\Cref{thm:nom-sh}) emphasizes  the property of finite supports being closed under intersection.
There exists a different description of the presheaves in the image of $\unsh$ as \emph{sheaves}, which is closer to the definition of a support. While intersection-preservation was the focus in \cite{p13,gh08}, the sheaf-based perspective was assumed in~\cite{s07} and in recent approaches on separation logic~\cite{l24,sbf25}, so it is worthwhile to connect the approaches.
This also allows us to conceptually include \emph{supported sets}~\cite{w23}, another abstract model of the notion of support.

We recall the basic definitions for (covariant) sheaves over a small category $\A$.
A \emph{cover} on \(C \in \A\) is a family of morphisms with domain \(C\).
A \emph{coverage} \(J\) on \A assigns to every \(C \in \A\) a collection \(J(C)\) of covers on \(C\). It is \emph{stable} if for every \(g \colon C \rightarrow D\) and every
 cover \(S \in J(C)\), there exists a cover \(T \in J(D)\) such that $T \cdot g$ factorizes through $S$, that is, for every \(h \cdot g \in T \cdot g\) there exist morphisms $f \in S$ and $g' \in \A$ with $h \cdot g = g' \cdot f$, see \eqref{eq:cover-stable}.\\
\begin{minipage}[c]{.3\columnwidth}
  \begin{equation}\label{eq:cover-stable}
  \begin{tikzcd}
    C \rar{ \forall g} \ar{d}[swap]{\exists f \in S} & D \dar{\forall h \in T} \\
    D' \rar{\exists g'} & E \\
  \end{tikzcd}
\end{equation}
\end{minipage}
\hfill%
\vspace{-10pt}\\
Given a presheaf \(F \in \PSh \A\), a \emph{matching family} for a cover \(S\) on~\(C\) is a family \((x_{f} \in F D)\) for \((f \colon C \rightarrow D) \in S\) such that \(g \cdot x_{f} = k \cdot x_{h}\) for all \(f, h \in S\) and morphisms \(g,k\) with \(gf = kh\). An \emph{amalgamation} of the family is an element \(x \in FC\) with \(f \cdot x = x_{f}\) for all \(f \in S\).
Let~\(J\) be stable a coverage and let \(F \in \PSh \A\) be a presheaf. Then \(F\) is \emph{$J$-separated} if every matching family has {at most one} amalgamation, and a \emph{$J$-sheaf} if every matching family has a {unique} amalgamation.
The full subcategory of sheaves in $\PSh \A$ is denoted by $\Sh({J})$.

\begin{expl}\label{ex:cov-I}
  For a contextual subcategory $\ctx \hookrightarrow \F$,
  we get a stable coverage \(\mathcal{I}_\ctx\) on \ctx whose only covers are singleton inclusions \(\{A \subseteq B\}\).
\end{expl}

We get the following correspondence between intersection-preserving presheaves and sheaves, capturing $\ctx=\I$~\cite[Example A2.1.11h]{j02} and $\ctx = \F$~\cite{gh08} uniformly:

\begin{thm}\label{thm:int-pres-sheaves}
  For every contextual category $\ctx\hookrightarrow \F$ we have \(\PShpb \ctx \simeq \Sh(\mathcal{I}_\ctx)\).
\end{thm}

Note that for $\A = \I,\, \F$ the sheaf condition matches precisely the support condition:
a presheaf $F$ is a sheaf if for every $y \in FB$ and every inclusion $A \subseteq B$ we have that whenever
\[\forall \pi, \sigma \colon  \pi|_A = \sigma|_A \Longrightarrow \pi \cdot y = \sigma \cdot y,\]
then $y$ is already supported by $A$:
there exists a unique $x \in FA$ with $(A \subseteq B) \cdot x = y$.

\subsection{Supported Sets}
\label{sec:supp-set}
A \emph{supported set}~\cite{w23} consists of a set $X$ with a \emph{support function} $s \colon X \rightarrow \pow_\f \names$, and a morphism of supported sets is a map $X \rightarrow X'$ satisfying $s'(f(x)) \subseteq s(x)$ for all $x \in X$, leading to a category $\Suppset$.
Supported sets give a foundation for \(\Nom\)-like categories in which the elements have a support, but the carriers may not be closed under name symmetries.
Their relation to nominal sets is therefore orthogonal to that between $\PSh \B$ and nominal sets (not only in this sense: while $\Nom$ is a Kleisli category over $\PSh \B$, it is monadic over $\Suppset$).
While supported sets have been tied to some presheaf categories under adjunctions~\cite{w23}, a presentation as a category of (pre-)sheaves has been missing so far; we fill this gap to fully integrate supported sets into the presheaf setting.

We denote by $\Idcat$ the discrete subcategory on $\pow_\f \names$ and by $\Inc$ the posetal category $\pow_\f \names$.
Since presheaves over $\Idcat$ are just $\pow_\f \names$-indexed families of sets, we have the standard equivalence \[\PSh \Idcat \simeq \Set \quot \pow_\f \names \simeq \Suppsetpres,\]
where $\Suppsetpres \incl \Suppset$ is the wide subcategory of support-preserving morphisms $f\colon X \rightarrow X'$ (that is, $s'(f(x)) = s(x)$ for all $x \in X$).
The inclusion $\iota \colon \Idcat \incl \Inc$ extends to a left adjoint $\Lan_\iota \colon \PSh \Idcat \rightarrow \PSh \Inc$ with induced monad $\mathcal{T}'$ on $\PSh \Idcat$.
This yields a monad $\mathcal{T}$ on $\Suppsetpres$
defined for $(X, s_X) \in \Suppsetpres$ by
\[ \mathcal{T} X = \{(x, A) \mid x \in X,\, s_X(x) \subseteq A\}\quad \text{with} \quad s_{\mathcal{T} X} \colon \mathcal{T} X \rightarrow \pow_\f \names, (x, A) \mapsto A.\]
Similar to nominal sets (\Cref{prop:nom-psh-square}) we obtain:
\begin{thm}\label{thm:supp-set-char}
  \begin{enumerate}
    \item  $\Suppset$ is the Kleisli category of \(\mathcal{T}\).
    \item Supported sets are equivalent to $\cap$-preserving presheaves: $\Suppset \simeq \PShpb \Inc$.
  \end{enumerate}
\end{thm}

In contrast to the other categories of inter\-sec\-tion-preserving presheaves considered to far, the category \(\PShpb \Inc \simeq \Suppset\) is \emph{not} a (Grothendieck) topos~\cite{w23}, but merely a \emph{quasitopos}, meaning that its subobject classifier is defined with respect to \emph{strong} (rather than arbitrary) monomorphisms.
However, it still fits into the sheaf picture as a category of \emph{separated objects} in a sheaf topos, which is always a quasitopos~\cite{bp91}.

Let $\mathcal{O}$ be the coverage on $\Inc$ whose covers of $A \in \Inc$ are finite families $K = (\iota_i \colon A \subseteq B_i)_{i \in F}$ such that $A = \bigcap_{i \in F} B_i$.
Recall from \Cref{ex:cov-I} that $I_{\Inc}$ is the coverage with singleton inclusions as covers.

\begin{prop}\label{prop:suppset-shv}
  A presheaf $F \in \PSh \Inc$ preserves intersections iff $F$ is a $\mathcal{O}$-sheaf and $\mathcal{I}_{\Inc}$-separated.
\end{prop}

\subsection{The Nominal Landscape}\label{sec:nominal-world}
We can summarize the results of this section as follows. Recall that we consider the index categories below, all of which have finite sets as objects as morphisms the following maps:

\medskip
\hspace{-5pt}\begin{tabular}{c c c c c c}
  $\Idcat$ & $\Inc$ & $\B$ & $\I$ & $\S$ & $\F$ \\
  identities & inclusions & bijections & injections & surjections & all
\end{tabular}

\begin{equation}
  \label{eq:index-cats}
  \begin{tikzcd}
    \Idcat \rar[hook]{} \dar[hook]{}   & \B \rar[hook]{} \dar[hook]{}    & \S \dar[hook]{}   \\
    \Inc \rar[hook]{}  & \I \rar[hook]{}  & \F
  \end{tikzcd}
\end{equation}

\noindent Their presheaf categories are related as follows:
  \begin{prop}[{see e.g.~\cite[Expl.\ 4.2.7(b)]{j02}}]
    For every bijective-on-objects functor \(G \colon \A \rightarrow \D\) the adjunction \(\PSh G \colon \PSh \A \dashv \PSh \D \cocolon G^{*}\) is monadic.
  \end{prop}

Combined with the literature we obtain a complete picture of the various categories of models for abstract syntax and substitution, and of the formal connections between them:
\newcommand{\dequiv}{\dar[phantom]{\rotatebox{90}{$\simeq $}}}
\[
  \begin{tikzcd}[sep=large]
    \Suppsetpres \dequiv \dar[phantom,xshift=-25pt]{\text{\small{Sec.\ \ref{sec:supp-set}}}} & \Nompres \dequiv \lar[swap,dashed]{} \dar[phantom,xshift=-25pt]{\text{\small{Sec.\ \ref{sec:supp-set}}}} & \Relevpres \dequiv \lar[swap,dashed]{} \dar[phantom,xshift=-25pt]{\text{\small{Sec.\ \ref{sec:relev}}}}  & \\
    \PSh \Idcat \dar[squiggly]{} \dar[phantom,xshift=-25pt]{\text{\small{Sec.\ \ref{sec:supp-set}}}}     & \PSh \B \lar[swap,dashed]{\text{\cite{f01}}} \dar[phantom,xshift=-13pt]{\text{\small{\cite{f01}}}}   \dar[squiggly]{}   & \PSh \S \dar[squiggly]{} \lar[swap,dashed]{\text{\cite{f01}}}   \dar[phantom,xshift=-13pt]{\text{\small{\cite{sz15}}}}      & \PSh \F \lar[dashed,swap]{\text{\cite{f01}}}  \\
    \PShpb \Inc \dequiv \dar[phantom,xshift=-25pt]{\text{\small{Sec.\ \ref{sec:supp-set}}}} & \PShpb \I \dequiv \lar[dashed]{} \dar[phantom,xshift=-13pt]{\text{\small{\cite{j02}}}}    & \PShpi \F \dar[phantom,xshift=-25pt]{\text{\small{Sec.\ \ref{sec:relev}}}} \ar[no head,pos =.15,swap]{dd}{\rotatebox{90}{$\mathlarger\sim$}}[pos=.82]{\rotatebox{90}{$\mathlarger\sim$}} \rar[hook]{} \lar[swap]{(1)} & \PShpb \F \dar[phantom,xshift=-13pt]{\text{\small{\cite{s07}}}}   \uar[hook]{} \dequiv  \\
    \Sh_{\mathsf{sep}}(\mathcal{I}, \Sh(\mathcal{O})) \dequiv \dar[phantom,xshift=-25pt]{\text{\small{Sec.\ \ref{sec:supp-set}}}} & \Sh \mathcal{I}_\I \dequiv \lar[dashed]{} \dar[phantom,xshift=-13pt]{\text{\small{\cite{p13}}}}       & \phantom{X} & \Sh \mathcal{I}_\F \dar[phantom,xshift=-13pt]{\text{\small{\cite{gh08}}}}  \dequiv \ar[dashed, swap,pos=.3,crossing over]{ll}{(2)}  \\
    \Suppset & \Nom \lar[dashed,swap]{(3)} & \Relev \lar[swap]{(1')} \rar[hook]{} & \ar[dashed,shiftarr={yshift={-13pt}}]{ll}{(4)} \Ren
  \end{tikzcd}
\]

Dashed arrows $(\dashrightarrow)$ denote monadic right adjoints, and squiggly arrows $(\rightsquigarrow)$ left adjoints into Kleisli categories.
The second row was investigated by Fiore and Menni~\cite{f01,fm05}, and the sheaf categories on $\I$ and $\F$ in Staton's thesis~\cite{s07}, where monadicity of (2) is proven.
An explicit description of the left adjoint (3) and its monadicity is given in by Wissmann~\cite[Def.\ 4.3]{w23} and Moerman and Rot~\cite[Def.\ 4]{m20} gave an explicit construction of (4).
While the functors (1) and (1') are not monadic they still have left adjoints since the left adjoint of (4) corestricts to relevance sets.

\section{Conclusion and Future Work}
\label{sec:concl-future-work}

We have shown how to uniformly derive the closed monoidal substitution structure for presheaf categories over
different types of (untyped) contexts, as well as for different types of nominal sets. In the nominal case, the substitution tensors fill a notable gap in the theory of nominal sets. In addition, we
have exposed new connections between nominal and presheaf models, extending the nominal
landscape. While this is a rather technical contribution to begin with, it is intended as a
starting point for further development.

All ingredients are now at hand to develop the initial semantics for binding signatures as introduced by Fiore et al.~\cite{ftp99} at the level of nominal sets, for example, the characterization of the nominal set of $\lambda$-terms (with its substitution structure) as an initial $H$-monoid for an endofunctor $H$.
A related application is the semantics of higher-order recursion schemes~\cite{amv11} in which \emph{infinite} $\lambda$-terms form the initial \emph{completely iterative} $H$-monoid.
It is not clear yet how straightforward these applications will be, as the substitution tensor we described on $\Nom$ is \emph{affine}:
terms substituted for different variables have disjoint supports.
If this poses problems, then a possible solution might be to use renaming sets or the ``captureful'' substitution tensor (\Cref{rem:nom-sub-cap}).
Overall, it is now possible to compare the different axiomatizations of substitutions as given in~\cite{ftp99} in the nominal setting, in particular single-variable substitution as studied in~\cite{gm08,g16,g09} with the multi-variable monoid based approach from~\cite{ftp99}.

We have instantiated our theory to presheaves over untyped contexts. In remains to study whether it also applies to more involved (e.g.~dependently or parametrically) \emph{typed} contexts, and how it integrates with existing approaches to typed nominal sets.

Fiore and Ranchod~\cite{fr25} have recently developed the theory of \emph{single-variable} substitution for the types of contexts ($\B$, $\I$, $\S$, $\F$) we consider.
As an application of the extended nominal landscape, one could try to transfer this theory to the appropriate nominal models.

Finally, we aim to apply our results to the recently introduced \emph{higher-order abstract GSOS} framework~\cite{gmstu23}. In the latter, the operational semantics of higher-order languages such as the $\lambda$-calculus are modeled via suitable bifunctors on presheaf categories, and their definition involves the internal hom of the substitution tensor. With substitution tensors on nominal-like sets now being available, a technically simplified description of higher-order operational semantics over these models should be in reach.

\bibliography{bibliography}

\clearpage

\appendix
\renewcommand\theequation{A.\arabic{equation}}
In this appendix we provide full proofs and additional details.

\section*{Details for \Cref{sec:mon-struct-from-act}}

\appendixproof{thm:gen-mon}
  We have the following situation:
  \begin{equation}
    \label{eq:sit}
    \begin{tikzcd}
      \V \rar{J} \drar{\y} & \A \dar[xshift=5pt]{\nu} \dar[phantom]{\dashv} \\
      & \PSh \V^\op \uar[xshift=-5pt]{L}
    \end{tikzcd}
  \end{equation}
  where we write abbreviate \[\nu C = \nu_J C = \A(J(-), C)\] and
  \[LX = \Lan_\y(J) X = \int^A XA \cdot JA.\]
  The outer triangles of \eqref{eq:sit} commute:
  We have \[L \y A = \Lan_\y J (\y A) \cong J A\] since $\y$ is fully faithful,
  and \[\nu (JA) = \A(J-, JA) \cong \y A\] since $J$ is fully faithful.
  Moreover, since $J$ is dense we have $L \nu \cong \Id$.
  Recall that the monoidal structure $\tl_J$ is given by
  \[C \tl_J D = \Lan_J((-) \tl D)C \cong \int^{A \in \V} \nu(C)A \cdot A \tl D.\]
  Since $J$ is fully faithful we have
  \[JA \tl_J C = \Lan_J((-) \tl C)(J A) \cong A \tl C.\]
  We also have by definition of $J$ a natural isomorphism
  \[J(A \otimes B) = (A \otimes B) \tl J_I \cong A \tl (B \tl J_I) = A \tl JB.\]

  \noindent \emph{Units.} The right unit isomorphism is given by
  \begin{align}
    &\; C \tl_J JI \notag \\
    \cong&\; \Lan_J((-) \tl JI)C \notag \\
    \cong&\; \Lan_J(J(- \otimes I))C \label{step:ul1} \\
    \cong&\; \Lan_J (J)C \label{step:ul2} \\
    \cong&\; C,\label{step:ul3}
  \end{align}
  where we used compatibility for \eqref{step:ul1}, the unit isomorphism $\Id_\V \cong (-) \otimes I$ for \eqref{step:ul2} and density for \eqref{step:ul3}.
  The left unit isomorphism uses the unit of the action:
  \[JI \tl_J D \cong I \tl D \cong D.\]

  \noindent \emph{Associativity.} The associativity isomorphism is constructed  via the coend calculations
  \begin{align}
    &\; C \tl_J (D \tl_J E) \notag \\
    = &\; \int^A \nu(C) A \cdot A \tl (D \tl_J E) \notag \\
    \cong&\; \int^A \nu(C) A \cdot (A \tl D) \tl_J E \label{step:ass1} \\
    = &\; \int^A \nu(C) A \cdot \big(\int^B \nu (A \tl D) B \cdot B \tl E\big) \notag \\
    \cong &\; \int^B \big(\int^A \nu(C) A \cdot \nu (A \tl D) B \big) \cdot B \tl E \label{step:ass2} \\
    \cong &\; \int^B \nu\big(\int^A \nu(C) A \cdot A \tl D \big)B \cdot B \tl E \label{step:ass3} \\
    \cong &\; (C \tl_J D) \tl_J E, \notag
  \end{align}
  where we use associativity as in \eqref{eq:J-ass} for \eqref{step:ass1}; ``Fubini'' for coends and their distributivity over tensors for \eqref{step:ass2}; and well-behavedness of $J$ for $(-) \tl D$ for \eqref{step:ass3}, as it is precisely step \eqref{eq:wb-step} in the following calculation:
  \begin{align}
    &\; \int^A \nu(C)A \cdot \nu(A \tl D)B \notag \\
    =&\; \int^A \A(JA, C) \cdot \A(JB, A \tl D) \notag \\
    \cong&\; \Lan_J (\A(JB, (-) \tl D))C \notag \\
    \cong&\; \A(JB, \Lan_J((-) \tl D)C) \label{eq:wb-step}\\
    \cong&\; \A(JB, \int^A \A(JA, C) \cdot A \tl D) \notag \\
    =&\; \nu(\int^A \nu(C)A \cdot A \tl D)B \notag \qedhere
  \end{align}


\detailsfor{rem:coherence}
Assume that the conditions of \Cref{thm:gen-mon} are satisfied.
We denote the coend cone of $C \tl_J D$ at $f \colon JA \rightarrow C$ by
\[\kappa_f \colon A \tl D \rightarrow  C \tl_J D = \int^{A}\A(JA, C) \cdot A \tl D.\]
Since $(A \tl C) \tl_J D \cong \int^B \A(JB, A \tl C) \cdot B \tl D$ is a coend, the isomorphism \aiso is induced by a family
  \begin{equation}
    \label{eq:a-comp}
    \aiso_f \colon B \tl D \rightarrow A \tl (C \tl_J D), \qquad \text{for } f \colon JB \rightarrow A \tl C
  \end{equation}
  natural in $D$ and satisfying
  $\aiso_{f \cdot g} = \aiso_f  \cdot (g \tl D)$ {for all } $g \colon B' \rightarrow B$.
  Recall that it induces $\alpha \colon (C \tl_J D) \tl_J E \rightarrow C \tl_J (D \tl_J E)$ by
  \[
    \begin{tikzcd}
      & \int^B \A(JB, \int^A \A(JA, C) \cdot A \tl D) \cdot B \tl D \dar{\mathbf{wb}} \\
      B \tl E \urar[end anchor=west]{\kappa_{\kappa_f \cdot g}} \rar{\kappa_f \cdot \kappa_g} \drar[end anchor = west]{\kappa_f \cdot \aiso_g} & \int^A \A(JA, C) \cdot \int^B \A(JB, A \tl D) \cdot B \tl E \dar{\int^A \A(JA, C) \cdot \aiso}\\
      & \int^A \A(JA, C) \cdot A \tl \int^B \A(JB, D) \cdot B \tl E
    \end{tikzcd}
  \]
  for $f \colon JA \rightarrow C, g \colon JB \rightarrow A \tl D$ and where $\mathbf{wb}$ is the isomorphism due to the well-behavedness.

  Then $(\A, \tl_J, JI, \alpha, \lambda, \rho)$ is a monoidal category if the family $\aiso_{(-)}$ satsifies additional compatibility conditions:
  first, for the triangle condition $\aiso$  should extend the associator of the action:
  \begin{equation}
    \label{eq:a-cond-u}
    \begin{tikzcd}
       & A \tl (JI \tl_J D) \\
      (A \otimes I) \urar[xshift=-8pt]{\aiso_{\mathsf{as}}} \tl D \rar{\mathsf{as}} & A \tl (I \tl D) \uar{A \tl \kappa_{\id_I}},
    \end{tikzcd}
  \end{equation}
  where $\mathsf{as}$ is the associator of the action $\tl$.
  If this condition holds, then the triangle axiom holds by the action conditions.
  Second, $\aiso$ must satisfy the associativity condition
  \begin{equation}
    \label{eq:a-cond}
    \begin{tikzcd}
      B' \tl F \rar{\aiso_r} \dar{\aiso_{\aiso_q \cdot r}} & B \tl (E \tl_J F) \dar{\aiso_q} \\
      A \tl ((D \tl_J E) \tl_J F) \rar{A \tl \alpha} & A \tl (D \tl_J (E \tl_J F))
    \end{tikzcd}
  \end{equation}
  for all $A, B, B' \in \V$, $D, E, F \in \A$ and $q \colon J B \rightarrow A \tl D$ and $r \colon J B' \rightarrow B \tl E$.
  The pentagon identity now follows since \eqref{eq:a-cond} is just the restricted pentagon
  \[
    \begin{tikzcd}[column sep=tiny]
      ((A \tl D) \tl_J E) \tl_J F \rar{} \dar{} & (A \tl D) \tl_J (E \tl_J F) \drar{} & \\
      (A \tl (D \tl_J E) \tl_J F) \rar{} & A \tl ((D \tl_J E) \tl_J F) \rar{} & A \tl (D \tl_J (E \tl_J F))
    \end{tikzcd}
  \]
  reformulated in terms of $\aiso$.

\appendixproof{thm:gen-mon-clo}
  We assume that \eqref{eq:asm-ra} holds and show that $D \dwand E$ as above is an internal hom.
  Then we have a natural isomorphism
  \begin{align*}
    &\; \A(C \tl_J D, E) \\
    \cong &\; \A(\int^A \nu_J(C) A \cdot A \tl D, E) \\
    \cong &\; \int_A \Set(\nu_J(C)A, \A(A \tl D, E)) \\
    \cong &\; \Nat(\nu_J(C), \A((-) \tl D, E)) \\
    \cong &\; \Nat(\nu_J(C), \nu_J \Lan_J\A((-) \tl D, E)) \\
    \cong &\; \A(C, D \dwand E),
  \end{align*}
  where we used \eqref{eq:asm-ra} in the second-to-last step, and density of $J$ in the last step.

\section*{Details for \Cref{sec:psh-substitution}}

\appendixproof{prop:reindexing}
Note first that the reindexing map is given on fibers as follows:
\[
  \begin{tikzcd}
    f^{*}B \rar[equals]{} & \sum_{x}B_{f(x)} \rar[dashed]{r_{f}} & \sum_{y}B_{y} \rar[equals]{} & B \\
    & B_{f(x)} \uar{\iota_{x}} \urar[swap]{\iota_{f(x)}} & &
  \end{tikzcd}
\]

  Let \(\ctx\) be stable under pullbacks. Then \(\ctx\) is closed under products, for if \(f \colon A \rightarrow A', g \colon B \rightarrow B'\) are in \(\ctx\), then we consider the bundle \(\Delta B'\) over \(A'\), whose fibers are all equal to \(B'\), with base change \(f^{*}\Delta{B'}\).
  We then get the map \(f \times g\) as
  \[
    \begin{tikzcd}[column sep = small]
      \sum_{a} B \rar{\sum_{a} g} & \sum_{a} B' \rar[phantom]{\cong} & \sum_a f^{*}\Delta B'  \rar{r_{f}} & \sum_{a'} {B}    \\
      A \times B \uar[phantom]{\cong} \ar[dashed]{rrr}{f \times g} &  & & A' \times B' \uar[phantom]{\cong}
    \end{tikzcd}
  \]
  For the prime condition only the implication \(f + g \in \ctx \Rightarrow f, g \in \ctx\) has to be proven since we assume that \(\ctx\) is closed under sums.
  Here it suffices to show that if \(! + g \colon A + B \rightarrow 1 + B'\) is in \(\ctx\), then the morphism \(! \colon A \rightarrow 1\) is in \(\ctx\), as every map between (finite) sets can be written as a (finite) sum of constant ones.
  \[X \cong \sum_{y \in Y} h^{-1}(y) \rightarrow \sum_{y \in Y} \{y\} \cong Y. \]
  But \(! \colon A \rightarrow 1\) is the morphism \(r_{! + g}\) in
  \[
    \begin{tikzcd}
      A \rar{! = r_{! + g}} \dar[hook]{} & 1 \dar[hook]{} \\
      A + B \rar{! + g} & 1 + B',
    \end{tikzcd}
  \]
  which is in \(\ctx\) by assumption.

  For the other direction assume \(\ctx\) is closed under products and prime, and let \(B \rightarrow Y\) be a bundle over \(B\) with \(f \colon X \rightarrow Y\).
  Then we obtain \(r_{f} \in \ctx\) as follows:
  \[
    \begin{tikzcd}[column sep=small]
      f^{*}{B} \dar[phantom]{\cong} \ar[dashed]{rrr}{r_{f}} &  &  & {B} \dar[phantom]{\cong}  \\
      \coprod_{x}B_{f(x)} \rar[phantom]{\cong}  & \coprod_{y}f^{-1}(y) \times B_{y} \rar[every label/.style={yshift=8pt},font=\small]{\coprod_{y}f_{y} \times 1} & \coprod_{y} \{y\} \times B_{y} \rar[phantom]{\cong} & \coprod_{y} B_{y},
    \end{tikzcd}
  \]
  where \(f_{y} \colon f^{-1}(y) \rightarrow \{y\}\) are the restrictions of \(f\), which are in \(\ctx\) since it is prime. Note that the product condition is used in the middle arrow.

\appendixproof{thm:day-pow-bifunct}
   Given \(f \colon A' \rightarrow A\) in $\V$ and \(X \in \PSh{ \ctx }\), we construct the morphism \(f \tl X \colon A  \tl X \rightarrow A' \tl X\) by
  \[
    \begin{tikzcd}
      \int^{B \rightarrow A} \y B \cdot \prod_{a} X B_{a} \rar[dashed]{f\tl X} & \int^{B' \rightarrow A'} \y B' \cdot \prod_{a'} X B_{a'} \\
      \y B \cdot \prod_{a} XB_{a} \uar{\kappa_{B}} \rar{\y r_{f} \cdot p_{f}} & \y(f^{*}B) \cdot \prod_{a'} X (f^{*}B)_{a'} \uar{\kappa_{f^{*}B}}
    \end{tikzcd}
  \]
  The morphisms $r_f$ are in $\ctx$ by contextuality, and
  $f\tl X$ it is well-defined since the family \(\y r_f \cdot p_f\) over all \(B \rightarrow A\) is natural.

\appendixproof{lem:subst-props}
Most of these proofs are simple using the equivalence $A \tl X \cong X^{\oplus A}$, but we think it may also be instructive to use representation \eqref{eq:sub}.
\begin{enumerate}
  \item By the above isomorphism we have \[A \tl \y 1 \cong (\y 1)^{\oplus A} \cong \y (\underbrace{1 + \cdots + 1}_{A \text{ times}}) \cong \y A.\]
  \item The slice $\F / 0$ is a singleton, so \[X^{*0} = \int^{B \rightarrow 0} \y 0 \cdot \prod_{x \in 0}XB_{0} \cong \y 0 \cdot 1 \cong \y 0.\]
  \item Here we use extensiveness of the category \F we have
        \begin{align*}
          &(A + A') \tl X \\
          &\cong \int^{C \rightarrow A + A'} \y C \cdot \prod_{\alpha \in A + A'} XC_{\alpha} \\
          &\cong \int^{B \rightarrow A, B' \rightarrow A'} \y (B + B') \cdot \prod_{a \in A} XB_{a} \cdot \prod_{a' \in A'} XB'_{a'}  \\
          &\cong (A \tl X) * (A' \tl X)
        \end{align*}
  \item Under the equivalence \(\F / 1 \cong \F\) this is just a reformulation of \ref{eq:y}(3).
  \item
        Using that \(A \tl X \cong X^{\oplus A}\) we simply have by \iref{lem:subst-props}{sum}
        \[(A \times A') \tl X \cong (\coprod_{A} A') \tl X \cong  (A' \tl X)^{\oplus A} \cong A \tl A' \tl X. \]
  \item This follows from Items \ref{lem:subst-props:right-1} and \ref{lem:subst-props:prod}.
\end{enumerate}

\appendixproof{prop:day-power-lan}
  By a lengthy coend computation we have:
  \begin{align*}
    &\;\Lan_\iota(A \tl X) \\
    \cong&\; \int^{B \in \ctx} (A \tl X)B \cdot \y_{\ctx'}(\iota B)  \\
    \cong&\; \int^{B \in \ctx}\int^{C \rightarrow A} \prod_a X C_a \cdot \y_{\ctx}(C)B \cdot \y_{\ctx'}(\iota B)  \\
    \cong&\; \int^{C \rightarrow A} \prod_a X C_a \cdot \y_{\ctx'}(\iota C)  \\
    \cong&\; \int^{C \rightarrow A} \prod_a X C_a \cdot \y_{\ctx'}(\coprod_a \iota C_a) \\
    \cong&\; \int^{B' \rightarrow A}\int^{C \rightarrow A} \y_{C'}B' \cdot \prod_a \y_{C'}(\iota C_a)B_a' \cdot X C_a \\
    \cong&\; \int^{B' \rightarrow A} \y_{C'}B' \cdot \prod_a \big( \int^{C_a} \y_{C'}(\iota C_a)B_a' \cdot X C_a  \big) \\
    \cong&\; \int^{B' \rightarrow A} \y_{C'}B' \cdot \prod_a (\Lan_\iota \, X) B_a' \\
    \cong&\; \iota A \tl (\Lan_\iota\, X) &
  \end{align*}

\appendixproof{prop:sub-pres-day}
We compute
\begin{align}
  & (X \sub Z) \oplus (Y \sub Z) \notag \\
  \cong&\; \int^{A, B} \y(A + B) \cdot (X \sub Z)A \cdot (Y \sub Z)B  \notag \\
  \cong&\; \int^{A, B, C, D} \y(A + B) \cdot (XC \cdot (C \tl Z)A) \cdot (YD \cdot (D \tl Z)B) \notag \\
  \cong&\; \int^{A, B, C, D} \y A \oplus \y B \cdot (XC \cdot (C \tl Z)A) \cdot (YD \cdot (D \tl Z)B) \notag \\
  \cong&\; \int^{C, D} (C \tl Z) \oplus (D \tl C) \cdot XC \cdot YD   \label{sbd1}\\
  \cong&\; \int^{C, D} ((C + D) \tl Z) \cdot XC \cdot YD   \label{sbd2}\\
  \cong&\; \int^{A, C, D} \y(C+D)A \cdot (A \tl Z) \cdot XC \cdot YD   \label{sbd3}\\
  \cong&\; \int^{A} (X\oplus Y)A \cdot (A \tl Z)   \notag \\
  \cong&\; (X\oplus Y) \sub Z \notag \qedhere
\end{align}
where we use Yoneda twice for \eqref{sbd1}; \Cref{lem:subst-props} for \eqref{sbd2}; and Yoneda again for \eqref{sbd3}.

\appendixproof{thm:subst-monoidal}
  We apply \Cref{thm:gen-mon} to the following setting:
  \begin{itemize}
    \item $\V$ is the dual of the monoidal category $(\ctx, \times, 1)$; note that $\ctx$ is closed under $\times$ by \Cref{prop:reindexing}.
    \item $\A$ is the category $\PSh{ \ctx }$.
    \item $J=\y\colon \ctx^\op\to \PSh{ \ctx }$ is the Yoneda embedding.
    \item The action $\tl$ is that of \Cref{def:sub}; it is an action by \Cref{lem:subst-props}.
  \end{itemize}
We know by \Cref{expl:wb} that $\y$ is well-behaved.  Compatibility of the action with $\y$ holds by \Cref{lem:subst-props}, and associativity follows from \Cref{cor:conv-rep-sub} together with \Cref{prop:sub-pres-day}:
\[A \tl (X \sub Y) \cong \bigoplus_{a\in A} X \sub Y\cong (\bigoplus_{a\in A}X) \sub Y \cong (A \tl X) \sub Y.\]

Closedness and the description of $Y\dwand Z$ are immediate from \Cref{thm:gen-mon-clo}, using that $\nu_\y\cong \Id$ and $\Lan_\y\y\cong \Id$.

\appendixproof{thm:sub-pres}
Given \(X, Y \in \PSh{ \ctx }\) we compute
\begin{align}
  & (\Lan_\iota X) \sub (\Lan_\iota Y) \notag \\
  = & \int^{B' \in \ctx'} (\Lan_\iota X)B' \cdot (B' \tl (\Lan_\iota Y)) \notag \\
  \cong & \int^{B' \in \ctx', A \in \ctx} \y_{\ctx'}(\iota A, B') \cdot XA \cdot (B' \tl (\Lan_\iota Y))  \label{psh1}\\
  \cong & \int^{A \in \ctx}  XA \cdot (\iota A \tl (\Lan_\iota Y))  \label{psh2} \\
  \cong & \int^{A \in \ctx}  XA \cdot \Lan_\iota (A \tl Y) \label{psh3} \\
  \cong & \Lan_\iota (\int^{A \in \ctx}  XA \cdot (A \tl Y)) \label{psh4}  \\
  = & \Lan_\iota (X \sub Y) \notag
\end{align}
with the definition of $\Lan_\iota$ for \eqref{psh1}; Yoneda for \eqref{psh2}; {\Cref{prop:day-power-lan}} for \eqref{psh3}; {\(\Lan_\iota\) preserves colimits} for \eqref{psh4}.

We moreover have \(\Lan_\iota (\y_{\ctx} 1) = \y_{\ctx'} 1\) since \(\iota 1 = 1\) and \(\Lan_\iota \) lifts \(\iota\), see \Cref{rem:psh-functor}.

\appendixproof{prop:sub-monoids-char}
  Given a \(\sub\)-monoid \(M\) we define the monad $T_M$ by right action: \(T_{M} = (-) \sub M\). This monad preserves colimits as $(-) \sub M$ is a left adjoint, and it preserves Day convolution by \Cref{prop:sub-pres-day}.

  Given monad \(T\) on $\PSh{ \ctx }$ preserving colimits and $\oplus$, we define a \(\sub\)-monoid structure on \(M_T = T \y 1\).
  The unit is simply $e = \eta_{\y 1} \colon \y 1 \rightarrow T \y 1$.
  Note that since $T$ preserves $\oplus$ we have
  \begin{equation}
    \label{eq:T-pres-oplus}
    A \tl TX \cong (TX)^{\oplus A} \cong T(X^{\oplus A}) \cong T(A \tl X),
  \end{equation}
  and therefore
  \begin{align*}
    &\;(T \y 1) \sub (T \y 1) \\
    \cong&\; \int^{A} (T \y 1) A \cdot A \tl (T \y 1) \\
    \cong&\; T \big( \int^{A} (T \y 1) A \cdot A \tl \y 1 \big) & \text{ \(T\) preserves \(\int\) and \eqref{eq:T-pres-oplus}} \\
    \cong&\; T \big( \int^{A} (T \y 1) A  \cdot \y A \big) & \text{ \Cref{lem:subst-props}} \\
    \cong&\; TT \y 1 & \text{Yoneda.}
  \end{align*}
  We obtain the multiplication by composing this isomorphism with the multiplication \(\mu_{\y 1} \colon TT \y 1 \rightarrow T \y 1\)
  Using Yoneda and \Cref{lem:subst-props} we verify
  \[T_{M_{T}} X \cong \int^{A} XA \cdot A \tl (T \y 1) \cong T \big( \int^{A} XA \cdot A \tl \y 1 \big) \cong T X, \]
  and also \(M_{T_{M}} = T_{M} \y 1 = (\y 1) \sub M \cong M\). \qedhere

\appendixproof{prop:unif-sub-embed}
We show that the morphism is given by
  \begin{align*}
    \label{eq:unif-sub-embed}
    \int^{A, B}\!\!\!\!\ctx(A \times B, C) \cdot XA \cdot YB &\rightarrow  \int^{A}\!\!\! XA \cdot \int^{E \rightarrow A}\!\!\!\! \ctx(E, C) \cdot \prod_{a}YE_{a} \\
    \kappa_{A, B}(j, x, y) &\mapsto \kappa_{A}(x, \kappa_{\pi_{A}}(j, (y)_a)),
  \end{align*}
  where \(\pi_{A} \colon A \times B \rightarrow A\) is the product projection and \((y)_a \in \prod_a YB\) is the constant $A$-indexed family with value $y\in YB$.

We first have to show that $\varphi$ is well-defined:
Let $A, B \in \F$,  $j \colon A' \times B' \rightarrow C, x \in XA, y \in YB$ and $f \colon A \rightarrow A', g \colon B \rightarrow B'$.
First note that for constant bundles $\pi_{A'} \colon A' \times B' \rightarrow A'$ we have $f^{*}(A' \times B') = (A \times B' \rightarrow A)$, and $r_{f} = f \times 1$.
Then
\begin{align*}
  &  \varphi(\kappa_{A', B'}(j,  f \cdot x, g \cdot y)) \\
  =\; & \kappa_{A'}(f \cdot x, \kappa_{\pi_{A'}}(j, (g \cdot y)_{a})) \\
  =\; & \kappa_{A}(x, \kappa_{A \times B' \rightarrow A}(j \cdot (f \times 1), (g \cdot y)_{a})) \\
  =\; & \kappa_{A}(x, \kappa_{{A \times B \rightarrow A}}(j \cdot (f \times 1) \cdot (1 \times g), (y)_{a})) \\
  =\; & \varphi(\kappa_{A, B}(j \cdot (f \times g), x, y)).
\end{align*}
It is obvious that $\varphi$ is natural, so it remains to show that it is a morphism of monoidal structures.
We just give the proof for the associators (the proof for the unitors is simple), showing that the following diagram commutes for all presheaves $X, Y, Z \in \PSh \F$:
\begin{equation}
  \label{eq:ass-mor}
  \begin{tikzcd}
    (X \otimes Y) \otimes Y \rar{\alpha} \dar{\varphi \otimes 1} &  X \otimes (Y \otimes Z) \dar{1 \otimes \varphi} \\
    (X \sub Y) \otimes Y \dar{\varphi} & X \otimes (Y \sub Z) \dar{\varphi} \\
    (X \sub Y) \sub Z \rar{\alpha} & X \sub (Y \sub Z)
  \end{tikzcd}
\end{equation}
Recall for this that the associator
\begin{align*}
  &\; \int^{C, D} \y(C \times D)E \cdot (\int^{A, B} \y(A \times B) \cdot XA \cdot YB)C \cdot ZD \rightarrow\\
  &\; \int^{C', D'} \y(C' \times D' )E \cdot XC' \cdot (\int^{A', B'} \y(A' \times B') \cdot YA' \cdot ZB')D'
\end{align*}
for Day convolution $\otimes$ sends an element
\[[f \colon C \times D \rightarrow E, [g \colon A \times B \rightarrow C, x \in XA, y \in YB]_{A, B}, z \in ZD ]_{C, D}\]
to the element
\[[f(g \times 1) \colon A \times (B \times D) \rightarrow E, x, [1_{B \times D}, y, z]_{B, D}]_{A, B \times D},\]
where we now denote, to ease notation, the coend coprojections by $[-]_A = \kappa_A$ or just $[-]$, and the currying of a map $f \colon A \times B \rightarrow C$ at $a \in A$ by $f_a \colon B \rightarrow C$.
The associator
\[\int^{A} (\int^B XB \cdot Y^B)A \cdot Z^AE \rightarrow \int^{B} XB \cdot (\int^{A} YA \cdot Z^A)^B E\]
is defined by
\[[[x, (y_b)_b]_B, (z_a)_a]_A \mapsto [x, ([y_b, (z_a)_a]_{A})_b]_B.\]
So for the lower-left path of \eqref{eq:ass-mor}
we have
\begin{align*}
  &\; \alpha\varphi(\varphi \otimes 1)[f, [g, x, y], z] \\
  =&\; \alpha\varphi [f, [x, (g_a \cdot y)_a]_A, z] \\
  =&\; \alpha[[x, (g_a \cdot y)_a]_A, (f_c \cdot z)_c]_C \\
  =&\; [x, ([g_a \cdot y, (f_c \cdot z)_c]_C)_a]_A
\end{align*}
while for the upper-right path we have, writing $1_b \colon D \rightarrow B \times D, d \mapsto (b, d)$ that
\begin{align}
  &\; \varphi(1 \otimes \varphi)\alpha [f, [g, x, y, z]] \notag \\
  =&\; \varphi(1 \otimes \varphi)[f(g \times 1), x, [1, y, z]] \notag \\
  =&\; \varphi [f(g \times 1), x, [y, (1_b \cdot z)_b]_B] \notag \\
  =&\; [x, ((f(g \times 1))_a \cdot [y, (1_b \cdot z)_b])_a ]_A \notag \\
  =&\; [x, ([y, ((f(g \times 1))_a \cdot 1_b \cdot z)_b]_B)_a ]_A \label{step:unisub} \\
  =&\; [x, ([y, (f_{g_a(b)} \cdot z)_b]_B)_a ]_A \notag \\
  =&\; [x, ([g_a \cdot y, (f_c \cdot z)_c]_C)_a]_A \label{step:reindex}
\end{align}
where \eqref{step:unisub} is the definition of the action of the substitution tensor, and \eqref{step:reindex} is reindexing of substitutions.

\appendixproof{thm:com-mon}
  Let $\mathcal{T}$ be a finitary monad on $\Set$. Its restriction along $\F \hookrightarrow \Set$ yields the corresponding $\sub$-monoid $(T, m, e)$ in $\PSh \F$.
  Let $T_\otimes$ be the associated $\otimes$-monoid with the same carrier, but with multiplication induced by $\varphi$, viz.\ \(m_\otimes = m \cdot \varphi \colon T \otimes T \rightarrow T \sub T \rightarrow T\).
  Intuitively, the multiplication $m_\otimes$ can only collapse ``uniform'' terms.
  Day convolution $\otimes$ freely extends the monoidal structure $\otimes$ of $\ctx$, so it has a universal property~\cite{day70}: \emph{commutative} $\otimes$-monoids are precisely \emph{symmetric} lax \(\times\)-monoidal functors $(\F, \times) \rightarrow (\Set, \times)$.
  Therefore, $T_\otimes$ is commutative iff the following diagram commutes, where $\tau$ is the symmetry of $\times$ and the horizontal arrows are induced by the multiplication of $T$.
  \begin{equation}
    \label{eq:comm-mon}
    \begin{tikzcd}
      T_{\otimes}A \times T_{\otimes}B \rar{} \ar{d}[swap]{\tau} & T_{\otimes}(A \times B) \dar{T_{\otimes} \tau} \\
      T_{\otimes}B \times T_{\otimes}A \rar{} & T_{\otimes}(B \times A)
    \end{tikzcd}
  \end{equation}
 But commutativity of \eqref{eq:comm-mon} means precisely that the finitary set monad $T$ is commutative.

\section*{Details for \Cref{sec:subst-nomin-sets}}
\appendixproof{prop:day-fresh}

  Denoting the coend cone for the Day convolution on $\PSh \I$ by
  \[\kappa \colon \I(A + B, C) \times FA \times GB \rightarrow (F \oplus G)C,\]
  the isomorphism for $\unsh$ is given as follows for nominal sets \(X, Y\) by
  \begin{align*}
    (\unsh X \oplus \unsh Y)(C)  &\cong \unsh (X * Y)C \\
    \kappa_{A,B}(j \colon A + B \mono C, x, y) &\mapsto (j|_{A} \cdot x, j|_{B} \cdot y) \\
    \kappa_{\supp x \cup \supp y}(\supp x \cup \supp y \subseteq C, x, y) &\mapsfrom (x, y)
  \end{align*}
  In the direction \(\PSh \I \rightarrow \Nom\) we have for sheaves \(F, G\):
  \begin{align*}
    \sh F * \sh G &\cong \sh (F \oplus G) \\
    ([A, x \in FA], [B, x \in FB]) &\mapsto [A \cup B, \kappa_{A, B}(\id_{A \cup B}, x, y)] \\
    ([j[A], j|_{A} \cdot x], [j[B], j|_{B} \cdot y]) &\mapsfrom [C, \kappa_{A, B}(j \colon A + B \mono C, x, y)] \\
  \end{align*}

\appendixproof{thm:nom-sub-pres}
For the purpose of this proof, suppose that $X\sub Y$ is \emph{defined} by \eqref{eq:sub-pres}. Since $\nu_\psi \cong I_*$ (\Cref{lem:nom-nerve}), our task is to show
  \[\int^{A \in \I} I_* X A \cdot Y^{* A} \cong X \sub Y.\] The isomorphism sends an element $\kappa_A(x \in I_* X A, \gamma \in Y^{*A})$ of the coend on the left to the equivalence class $x[\gamma|_{\supp x}]$ on the right,
  and the equivalence relation $\sim$ of \eqref{eq:sub-pres} corresponds precisely to the equivalence relation of the coend.

\appendixproof{thm:nom-sub}
  (1):
  Every representable $\y A \in \PSh \I$ preserves pullbacks, so the Yoneda embedding $\y \colon \I^\op \rightarrow \PSh I$ factorizes through $\PShpb \I \subseteq \PSh \I$, and we have $I_* \cdot \psi = \y$.
  So $\psi$ is fully faithful and dense since $\y$ is.
  The well-behavedness condition is satisfied by \Cref{lem:wb-psh-subcat} since $\Lan_\y (\iota Y^{*(-)}) \colon \PSh \I \rightarrow \PSh \I$ corestricts to $\Nom$: The image of $X$ is precisely $X \sub Y$.

  (2 + 3): The isomorphism
  \[\psi A \tl \psi B = (\V^{*B})^{*A} \cong \V^{*(A \times B)} \cong \psi(A \times B)\]
  is clear, and $A \tl (X \tl_\psi Y) \cong (A \tl X) \tl_\psi$ is given by
  \begin{align*}
    (X \sub Y)^{*A} &\cong X^{*A} \sub Y \\
    (\lambda a.\ x_a[\gamma_a]) &\leftrightarrow (\lambda a.\ x_a)[\lambda (b \in \supp x_a).\ \gamma_a(b)]. \qedhere
  \end{align*}

\appendixproof{prop:dwand}
  \begin{enumerate}
    \item The action is well-defined, i.e., if $f \in Y \dwand Z$ then $\pi \cdot f \in Y \dwand Z$: let $f = f' \cdot \pi_A \colon X^{*\V} \rightarrow X^{*A} \rightarrow Y$ be a finite reduction of $f$, then the following diagram commutes:
          \[
          \begin{tikzcd}
            {Y^{*\V}} \rar{Y^{* \pi}} \dar{p_{\pi A}} \ar[shiftarr={yshift=15pt}]{rr}{\pi \cdot f} & Y^{* \V}\rar{f} \dar{p_A} & Z \\
            Y^{* \pi A} \rar{Y^{* \pi}} &  Y^{* A} \urar[swap]{f'} &
          \end{tikzcd}
          \]
          This shows that $\pi \cdot f$ is finitely reducible.
          It is easy to see that this yields an action.
          The above diagram also shows that for a finitely reducible $f  = f' \cdot p_A$ the finite set $A$ supports $f$: if $\pi|_A = \id$, then the lower path is equal to $f$ and the upper path to $\pi \cdot f$.

          We also need the other direction, which is subtle:
          we show that for $f \in Y \dwand Z$, if $A$ supports $f$ then $f$ factorizes through $p_A$.
          Let $\gamma, \delta \in Y^{*\V}$ such that $\gamma|_A = \delta|_A$, we have  to show $f(\gamma) = f(\delta)$.
          Since $f \in Y \dwand Z$ there exists a $B \subseteq_\f \V$ through which $f$ factors, so whenever
          \begin{equation}
            \label{eq:B-res}
            \gamma'|_B = \delta'|_B \text{ for } \gamma', \delta' \in Y^{*\V}, \text{ then } f(\gamma') = f(\delta').
          \end{equation}
          For simplicity assume $B \setminus A = \{b\}$, but the argument is the same for larger differences.
          If $\gamma(b) = \delta(b)$ we are done, since then $\gamma, \delta$ agree on all of $B$ and so $f(\gamma) = f(\delta)$, hence we assume $\gamma(b) \ne \delta(b)$.
          Let $a \in A \setminus B$ and $\gamma(a) = y_a = \delta(a)$.
          Given $\phi \in Y^{*C}$ we abuse notation and write $f(\phi)$ for $f(\phi')$ for some $\phi' \in Y^{*\V}$ extending $\phi$ that agrees with $\phi$ on $C \cup B$.
          Pick $c$ such that $\gamma(c)$ is fresh for $\delta(b)$ and some $y$ fresh for $\delta(b), \gamma(c)$.
          Then
          \begin{align}
            &\; f(\gamma) \notag \\
            =&\; f(\gamma \tr b c ) \label{step:ext1}\\
            =&\; f[a \mapsto \gamma(a), b \mapsto \gamma(c), c \mapsto \gamma(b)] \notag\\
            =&\; f[a \mapsto \delta(b), b \mapsto \gamma(c), c \mapsto y] \notag\\
            =&\; f[a \mapsto y, b \mapsto \gamma(c), c \mapsto \delta(b)] \notag\\
            =&\; f[a \mapsto y, b \mapsto \delta(b), c \mapsto \gamma(c)] \label{step:ext4}\\
            =&\; f(\delta), \notag
          \end{align}
          where steps \eqref{step:ext1} and \eqref{step:ext4} are due to $f$ being supported by $A$, and the other equalities are due to \eqref{eq:B-res}.

    \item The isomorphism is given by
          \begin{align*}
            \omega \colon \Nom(Y^{*A}, Z) &\cong I_*(Y \dwand Z) A \cocolon \omega^{-1} \\
            \alpha_A &\mapsto (\lambda (\gamma \in Y^{*\V}.\ \alpha(\gamma|_A))) \\
            (\lambda (\gamma \in Y^{*A}).\ {f}(\hat{\gamma})) &\mapsfrom f,
          \end{align*}
          such that $\hat{\gamma} \in Y^{*\V}$ is any map with $\hat{\gamma}|_{A} = \gamma|_{A}$.
          Since $f \in I_*(Y \dwand  Z)A$, it is supported by $A$, so this definition is independent of the choice of $\hat{\gamma}$.
  \end{enumerate}

\appendixproof{prop:nom-univ-sub-from-presh}
We show
\begin{equation}
  \label{eq:515todo}
  X \otimes Y \cong \sh(\unsh X \otimes \unsh Y) \cong \int^{A, B} \V^{*(A \times B)} \cdot \unsh X A \cdot \unsh Y B
\end{equation}
since $\sh$ preserves colimits, and $\sh \cdot \y \cong \psi$.
We denote the coend cone for the right side by
  \[\kappa_{A, B} \colon \V^{*(A, B)} \cdot \unsh X A \cdot \unsh Y B \rightarrow \int^{A, B} \V^{*(A \times B)} \cdot \unsh X A \cdot \unsh Y B,\]
  and we may omit the subscripts of $\kappa$ for readibility reasons.
  From left-to-right, the isomorphism \Cref{eq:515todo} is
  \begin{align*}
    \varphi \colon x[y] \mapsto \kappa_{\supp x \cup  \supp y}((a, b) \mapsto j_a(b), x, y),
  \end{align*}
  where $y$ is a fresh element in the orbit of all the $y_a$ used in $x[y]$,
  and for $a \in \supp x$ the $j_a \colon \supp y \cong \supp y_a$ satisfy $y_a = j_a \cdot y$.

  First note that this definition is independent of the choice of the $j_a$ up to the condition $y_a = j_a \cdot y$.
  This makes it independent in the choice of $y$ as well: if $y'$ satisfies the same conditions then there exists $pi$ such that $\pi \cdot y' = y$
  since they are in the same orbit, and we can choose $j'_a(b' \in \supp y') = j_a(\pi b)$.
  We then have
  \begin{align*}
    &\; \kappa((a, b) \mapsto j_a(b), x, y) \\
    =&\; \kappa((a, b) \mapsto j_a(b), x, \pi \cdot y')\\
    =&\; \kappa((a, b') \mapsto j_a(\pi b'), x, \pi \cdot y')\\
    =&\; \kappa((a, b') \mapsto j'_a(b'), x, \pi \cdot y')\\
  \end{align*}

  From right-to-left, it suffices to defin the isomorphism \Cref{eq:515todo}
  for $A = \supp x, B = \supp y$, and then
  \[\varphi^{-1} \colon \kappa_{A, B}(f \in \V^{*(A \times B)}, x, y) \mapsto x[(a \in supp x) \mapsto f_a \cdot y], \]
  where $f_a(b)= f(a, b)$.
  It is well-defined since
  \begin{align*}
    &\; \varphi^{-1}(\kappa(j \cdot (k \times l), x, y)) \\
    =&\; x[a \mapsto (j \cdot (k \times l))_a \cdot y] \\
    =&\; x[(a \mapsto j_a \cdot l)_a\cdot k \cdot y] \\
    =&\; (k \cdot x)[a \mapsto j_a \cdot l \cdot y] \\
    =&\; \varphi^{-1}(\kappa(j, k \cdot x, l \cdot y))
  \end{align*}
  It is routine to verify that these are indeed inverses.

\appendixproof{prop:nom-unif-sub-mon}
By \Cref{prop:nom-univ-sub-from-presh} we immediately have
\[\unsh(X \otimes Y) \cong \unsh \sh (\unsh X \otimes \unsh Y) \cong \unsh X \otimes \unsh Y.\]

  For the other direction we use \Cref{thm:day-conv-pres}.
  The functor \[\psi \colon (\V^\op, \times, 1) \rightarrow (\Nom, \otimes, \V)\] is obviously strong monoidal since $\V^{* A} \otimes \V^{* B} \cong \V^{*(A \times B)}$, and
  $\otimes$ on $\Nom$ preserves colimits in both variables:
  the inverse to the map \[\colim_i X \otimes Y_i \rightarrow X \otimes \colim_i Y_i\]
  sends $x[y] = x[a \mapsto j_a \cdot \kappa_i(y)]$ to $\kappa_i(x[a \mapsto j_a \cdot y])$.
  Therefore the extension of $\psi$, which by \Cref{lem:nom-nerve} is given by $\sh$, is a strong monoidal functor.

\section*{Details for \Cref{sec:taxon-nomin-sets}}

\appendixproof{lem:coprod-supp}
Clearly $\supp \kappa_A(x) \subseteq A$, so we show the other inclusion.
Suppose $a \in A \setminus \kappa_A(x)$, then for every fresh $b$ we have
\[\tr a b \cdot \kappa_A(x) = \kappa_{A - \{a\} \cup \{b\}}(\tr a b \cdot x) \ne \kappa_A(x) \]
since $A - \{a\} \cup \{b\} \ne A$, so $\tr a b \cdot \kappa_A(x)$ and $\kappa_A(x)$ are not in the same coproduct component.

\appendixproof{prop:nom-psh-square}
  \begin{enumerate}
    \item We have
          \[\Lan{ \iota }(F)B = \int^{A \in \B}\I(\iota A, B) \cdot FA,\]
          but every $\kappa_A(f \colon A \mono B, x \in FA)$ has the unique representative $\kappa_{f[A]}(\iota, f' \cdot x)$ for the image factorization \[f = \iota \cdot f' \colon A \cong f[A] \incl B\]
    \item We have
          \[RX = \coprod_{A \in \B}(\iota^* \unsh X)A = \coprod_{A \in \B}\{x \in X \mid \supp x \subseteq A\}.\]
    \item The inclusion \(i \colon \Nompres \rightarrow \Nom\) is a left adjoint to $R$ since there is an obvious natural correspondence between morphisms $f \in \Nom(iX, Y)$ and morphisms $\hat{f} \in \Nompres(X, RY)$ with $\hat{f}(x) = (f(x), \supp x)$.
    \item
  For $(x, A) \in RX$ we clearly have $\supp (x, A) = A$.
  So for a support-preserving map $f \in \Nompres(X, RY)$ with $f(x) = (\hat{f}(x), A)$ we have  \[A = \supp (\hat{f}(x), A) = \supp f(x) = \supp x.\]
  Therefore the maps
  \begin{align*}
    \Nompres(X, RiY) &\rightarrow \Nom(X, Y) \\
    f &\mapsto (x \mapsto \hat{f}(x)) \\
    (x \mapsto (g(x), A)) &\mapsfrom g
  \end{align*}
  are bijections.
  \end{enumerate}

\appendixproof{prop:mon-desc}
Using the descriptions $\mathcal{T} X = \{(x, A) \mid \supp x \subseteq A\}$ and $1_= = \pow_\f \V$, the isomorphism is the obvious one:
\begin{align}
  \mathcal{T} X  &\cong 1_= * X \label{eq:writer-iso}\\
  (x, A) &\mapsto (A \setminus \supp x, x) \notag \\
  (x, \supp x \cup B) &\mapsfrom (B, x). \notag
\end{align}
The unit of $\mathcal{T}$ is given by $X \rightarrow \mathcal{T} X, x \mapsto (x, \supp x)$, and its multiplication by $\mathcal{T} \mathcal{T} X \rightarrow X, ((x, A), B) \mapsto (x, A \cup B)$.
The unit of the writer monad $1_= * X$ is $X \rightarrow 1_= * X, x \mapsto (\emptyset, x)$
and the multiplication $1_= * (1_= * X) \rightarrow 1_= * X, (A, (B, x)) \mapsto (A \cup B, x)$.
It is easy to see that the isomorphisms \eqref{eq:writer-iso} are compatible with the monad structures.

\appendixproof{thm:relev}
\begin{enumerate}
  \item The proof of is analogous to that of \Cref{prop:spec-nompres-equiv}: the functors constituting this equivalence are given also given by
        \[\coprod^\S \colon \PSh \S \rightarrow \Relevpres, \qquad F \mapsto  \coprod_{A \subseteq_{\f} \V}FA \]
        and
        \[\mathcal{S}^\S \colon \Relevpres \rightarrow \PSh \S, \qquad  \mathcal{S}(X)(A) = \{x \in X \mid \supp x = A\}.  \]
        The action \(\rens \V \times \coprod^\S F \rightarrow \coprod^{\S} F\) is given by
        \[\rho \cdot \kappa_{A}(x) = \kappa_{\rho A}(\hat{\rho} \cdot x \in F(\rho A)),\]
        for \(\hat{\rho} \colon A \epi \rho A\).
        This makes $\coprod^\S F$ a relevance set with \(\supp \kappa_{A}(x) = A\).
        The functor structure of \(\mathcal{S} X\) is for \(\rho \colon A \epi B\) given by \(\rho \cdot (x \in \mathcal{S} X A) = \bar \rho \cdot x \in \mathcal{S} X B\),  where \(\bar \rho (a \in A) = \rho(a)\) and $\bar \rho(b \not \in A) = b$.
        Note that this definition only makes sense (i.e.\ \(\bar \rho \cdot x \in \mathcal{S} X B\)) because \(X\) is a relevance set and therefore \[\supp (\bar \rho \cdot x) = \bar \rho \cdot \supp x = \bar \rho \cdot A = B.\]
  \item Note that the monad $\bar{\mathcal{I}}$ does the same as the monad $\mathcal{I}$ on $\Nompres$, so in the Kleisli category the morphisms are allowed to (equivariantly) drop names from the support.
\end{enumerate}

\appendixproof{prop:preim-relev-equiv}
  (1) We show that for every preimage-preserving $F \in \PShpi \F$ the renaming set $\sh F$ is a relevance set.
  This is equivalent to showing $\rho \cdot \supp x \subseteq \supp \rho \cdot x$ for all $x \in \sh F$ and renamings $\rho$, since the inclusion $\supp \rho \cdot x \subseteq \rho \cdot \supp x$ holds for every renaming set.
  Since $F$ preserves preimages we have the following pullback of sets:
  \[
    \begin{tikzcd}
      F(\rho^{-1}[\supp \rho \cdot x]) \rar{} \dar[tail]{} & F(\supp \rho \cdot x) \dar[tail]{} \\
      F(\supp x) \rar{F \rho} & F(\rho \cdot \supp x).
    \end{tikzcd}
  \]
  In particular, since $x \in F(\supp x)$ and $\rho \cdot x \in F(\supp \rho \cdot x)$ we have \[x \in F(\rho^{-1}[\supp \rho \cdot x]) \subseteq F(\supp x),\] so $\supp x = \rho^{-1}[\supp \rho \cdot x]$ since $\supp x$ is the least support of $x$.
  This yields the desired inequality via
  \[\rho \cdot \supp x = \rho[\supp x] = \rho[\rho^{-1}[\supp \rho \cdot x]] \subseteq \supp \rho \cdot x. \]

  (2) Given a relevance set $X$ the presheaf $\unsh X$ preserves preimages:
  given $\rho \colon A \rightarrow B$ and $U \subseteq B$ we have to show that
  \[
    \begin{tikzcd}
      \unsh X(\rho^{-1} U) \rar{} \dar[tail]{} & \unsh X U \dar[tail]{} \\
      \unsh X A \rar{\unsh X \rho} & \unsh X B
    \end{tikzcd}
  \]
  is a pullback.
  This means that if $x \in \unsh X$ is supported by $A$ and $\rho \cdot x$ is supported by $U$, then $x$ is supported by $\rho^{-1}U$.
  As $X$ is a relevance set we have
  \[\rho [ \supp x ] \subseteq \supp \rho \cdot x \subseteq U,\] which is equivalent to $\supp x \subseteq \rho^{-1} U$, and we are done.

\appendixproof{thm:es-im-finit}

  Since $i = j \cdot \iota \colon \S \rightarrow \F \rightarrow \Set$  we have $\Lan_{i} \cong \Lan_j \cdot \Lan_i$.
  The equivalence between finitary endofunctors on $\Set$ and $\PSh \F$ is given by $\Lan_j$, whose pseudo-inverse is restriction $j^* \colon [\Set,\Set]_{\mathsf{finit}} \colon \PSh \F$.
  Since filtered colimits commute with finite limits, a finitary endofunctor $F \in [\Set, \Set]_\mathsf{finit}$, with $F \cong \Lan_j G$ for $G \in \PSh \F$, preserves preimages iff $G$ does.
  It therefore suffices to show that $G$ preserves preimages iff it is of the form $\Lan_\iota H$ for some $H \in \PSh \S$.

  The functor $\Lan_\iota$ is given for $H \in \PSh \S$ and $B \in \F$ by
  \begin{align}
    (\Lan_\iota)(H)(B) = \int^{A \in \S}\F(\iota A, B) \cdot HA &\cong \coprod_{A \subseteq B} HA \label{eq:psh-iota-s} \\
    [f, x] &\mapsto \iota_{f[A]}(f \cdot x) \notag \\
    [A \subseteq B, x] & \mapsfrom \iota_A(x). \notag
  \end{align}

  It is easy to see that every $\Lan_\iota H$ preserves preimages:
  In the square
  \[
    \begin{tikzcd}
      {\coprod_{A \subseteq f^{-1}U}H(f^{-1}U)} \rar{} \dar[tail]{} & {\coprod_{A' \subseteq U} HA'} \dar[tail]{} \\
      {\coprod_{A \subseteq B}HA} \rar{\Lan_\iota H f} & {\coprod_{A' \subseteq B'} HA'}
    \end{tikzcd}
  \]
  we have $\Lan_\iota Hf (\iota_A(x)) = \iota_{fA}(f \cdot x) \in \Lan_\iota H U$ iff $f A \subseteq U$ iff $A \subseteq f^{-1} U$, so $\iota_A(x) \in \Lan_\iota H (f^{-1}U)$.

  It remains to show that if $G \in \PSh \F$ preserves preimages then it is of the form $\Lan_\iota H$ for some $H \in \PSh \S$.

  Here we can use \Cref{prop:preim-relev-equiv}: we know that $\unsh G$ is a relevance set, so we have
  \[(\Lan_\iota)\mathcal{S}^\S \unsh G \cong \sh \unsh G \cong G.  \]

\appendixproof{prop:day-relev}
  Day convolution in \(\PSh \S\) is given by the coend
  \[(X \oplus Y)(C) = \int^{A, B \in \S} \S(A + B, C) \times XA \times YB\]
  whose coend cowedge we denote by \(\alpha\).
  Under the equivalence  from \Cref{thm:relev} Day convolution is transported to the monoidal structure in \Relevpres given by the cartesian product on their underlying sets, as witnessed by the isomorphism
  \begin{align*}
    \coprod^{\S} (X \oplus Y) &\cong \coprod^\S X \times \coprod^\S Y \\
    \iota_C(\kappa_{A,B}( j, x \in XA, y \in YB)) &\mapsto (\iota_{jA}(j_{A} \cdot x), \iota_{jB}( j_{B} \cdot y)) \\
    \kappa_{A, B}(A + B \epi A \cup B, x, y)  &\mapsfrom (\iota_{A}(x), \iota_{B}(y))
  \end{align*}
  where \(j = [j_{A}, j_{B}] \colon A + B \epi C\).

  This also makes the inclusion \(\Relevpres \rightarrow \Relev\) strict monoidal as products in $\Relev$ are taken in $\Ren$.

\detailsfor{rem:sub-relev}
It is easy to see that given two relevance sets $X, Y$ their substitution product $X \sub Y$ is a relevance set as well: given $x[\gamma] \in X \sub Y$ and $\rho \in \rens$ we have
\begin{align*}
  &\; \supp \rho \cdot x[\gamma] \\
  =&\; \bigcup_{a \in \supp x} \supp \rho \cdot \gamma(a) \\
  =&\; \bigcup_{a \in \supp x} \rho[\supp \gamma(a)] \\
  =&\; \rho[\bigcup_{a \in \supp x} \supp \gamma(a)] \\
  =&\; \rho[\supp x[\gamma]].
\end{align*}

However, the internal hom $\dwand$ of renaming sets does not restrict to relevance sets:
  Consider the discrete relevance set $Y = Z = 2$ and the element $f \in Y \dwand Z$ given by $f(\gamma) = 1$ if $\gamma(a) = \gamma(b)$, otherwise $f(\gamma) = 0$.
  For the renaming $\rho = [a \mapsto c, b \mapsto c]$ the map $\rho \cdot f$ is constant $1$, and thus has empty support.
  Therefore \[\emptyset = \supp (\rho \cdot f) \subsetneq \rho[\supp f] = \{c\},\]
  so $Y \dwand Z$ is not a relevance set.

\detailsfor{ex:cov-I}
We show that for a contextual subcategory $\ctx \hookrightarrow \F$ the coverage whose covers are singleton inclusions $\{A \subseteq B\}$ is stable:

  Given a cover \(S = \{C \subseteq D' = C + A\}\) of \(C\) and a \ctx-morphism \(g \in \ctx(C, D)\), take the cover $T = \{D \subseteq D + A' \}$ of $D$, where $A'$ is disjoint from $D$ and $\pi \colon A \cong A'$:
  \begin{equation}
  \begin{tikzcd}
    C \rar{ \forall g} \dar[swap, phantom]{\rotatebox{-90}{$\subseteq$}} & D \dar[phantom]{\rotatebox{-90}{$\subseteq$}} \\
    \mathllap{D' = \;}C + A \rar[dashed]{g + \pi} & D + A' \mathrlap{\;= E}
  \end{tikzcd}
\end{equation}
  Note that the inclusion $D \subseteq D \cup A'$ is in $\ctx$ by the prime property (\Cref{prop:reindexing}): $\id_C + i_A \colon C + \emptyset \rightarrow C + A$ is in $\ctx$, so $\id_D + i_A$ is as well.
  On $\ctx = \B, \S$ the only inclusions are identities, and so the sheaf condition for these coverages is trivial.

\appendixproof{thm:int-pres-sheaves}
  Under these conditions the proof from the Elephant~\cite[Example A2.1.11h]{j02} for $\ctx = \I$ generalizes, but for self-containment we recall the complete proof here.

  Note that presheaf $F \in \PSh \ctx$ is a sheaf iff for all inclusions $\iota \colon A \incl B$ the morphism $F \iota \colon FA \rightarrow FB$ factorizes as
  $FA \cong K \incl FB$, where
  \[\{y \in FB \mid \forall g, h \colon g \iota = h \iota \Longrightarrow g \cdot y = h \cdot y\}.\]

  First, let $F$ preserves intersections.
  It then also preserves monos, so $F\iota$ is injective.
  Now let $y \in K$.
  Since $\iota = 1 + i \colon A = A + 0 \rightarrow A + A' = B$ and $\ctx$ is an prime we have $i \in \ctx$.
  We can therefore express $A$ as the intersection of $g = 1_{A} + 1_{A'} + i \colon B \rightarrow A + A' $ and $h = 1_{A} + i + 1_{A'} \colon B \rightarrow A + A' + A'$.
  \[
    \begin{tikzcd}
      A \rar[hook]{\iota} \dar[hook]{\iota} & B \dar[hook]{g} \\
      B \rar[hook]{h} & A + A' + A'.
    \end{tikzcd}
  \]
  We have $g \iota = h \iota$, so, as $y \in K$, also $g \cdot y = h \cdot y$.
  Since $F$ preserves the intersection of $g, h$ there exists a unique $x \in FA$ with $\iota \cdot x = y$.

  Second, let $F$ be a sheaf, we show it preserves the intersection of $A, A' \subseteq C$.
  Let $y \in FA, y' \in FA'$ such that they are sent to the same element of $FC$.
  Given $g, h \colon A \rightarrow B$ that agree on $A \cap A'$, define $g', h' \colon C = A + P \rightarrow B + P$ via $g, h$, respectively:
  \[
    \begin{tikzcd}
      A \cap A' \rar[hook]{} \dar[hook]{} &  A \dar[hook]{} \rar[yshift=+3pt]{g} \rar[swap,yshift=-3pt]{h} & B \dar[hook]{} \\
      A' \rar[hook]{} & C \rar[yshift=+3pt]{g'} \rar[swap,yshift=-3pt]{h'} & B + P
    \end{tikzcd}
  \]
  In the right square of this diagram the upper and lower paths commute, so $g', h'$ agree on $A'$.
  Then we have
  \begin{align*}
    &\;(B \subseteq B + P) \cdot g \cdot y \\
    = &\; g' \cdot (A \subseteq C) \cdot y \\
    = &\; g' \cdot (A' \subseteq C) \cdot y' \\
    = &\; h' \cdot (A' \subseteq C) \cdot y' \\
    = &\; (B \subseteq B+ P) \cdot h \cdot y.
  \end{align*}
  and thus also $g \cdot y = h \cdot y$ since $F(B \subseteq B + P)$ is mono.
  Therefore there exists a unique $x \in FA$ with $(A \cap A' \subseteq A) \cdot x = y$.
  We also have
  \begin{align*}
    &\; (A' \subseteq C) \cdot (A \cap A' \subseteq A') \cdot x \\
    = &\; (A \subseteq C) \cdot (A \cap A' \subseteq A) \cdot x \\
    = &\; (A \subseteq C)\cdot y \\
    =&\; (A' \subseteq C) \cdot y',
  \end{align*}
  so $(A \cap A' \subseteq A') \cdot x = y'$ as $F(A \cap A' \subseteq A')$
  is mono as well.

\appendixproof{thm:supp-set-char}
  \begin{enumerate}
    \item This is just the standard equivalence between families of sets and slices.
    \item Analogous to nominal sets.
    \item The equivalence is analogous to that of $\Nom$ and $\PShpb \I$:
          The functor \(\Lambda \colon \PShpb \Inc \rightarrow \Suppset\) sends a sheaf \(F\) to the supported set \(\Lambda F\) whose underlying set is \(\colim F\); since \(\Inc\) is directed this is given by
  \(\coprod_{A \subseteq \V} FA / \sim\), where \(\kappa_{A}(x) = \kappa_{B}(y)\) if \[(A \subseteq C) \cdot x = (B \subseteq C) \cdot x \text{ for some \(C \supset A, B\).}\]
  The map \(s \colon \Lambda F \rightarrow \powf \V\) is given by \[s(\kappa_{A}(x)) = \bigcap_{\kappa_{B}(y) = \kappa_{A}(x)} B,\]
  it is well-defined by definition and since \(F\) preserves intersections it makes sense: if \(\kappa_{A}(x) = \kappa_{B}(y)\) there exists a unique \(z \in F(A \cap B)\) with \((A \cap B \subseteq A) \cdot z = x\) and \((A \cap B \subseteq B) \cdot z = y\), so every equivalence class \(\kappa_{A}(x)\) has a ``smallest representative'' \(\kappa_{s(x)}(x_{s})\) with \((s(x) \subseteq A) \cdot x_{s} = x\).

  For \(\alpha \colon F \Rightarrow G\) the map \[\Lambda \alpha \colon \Lambda F \rightarrow \Lambda G, \quad \kappa_{A}(x) \mapsto \kappa_{A}(\alpha(x))\] is well-defined since \(\alpha\) is natural, and it is a morphism in \(\Suppset\):
  Let \(s(\kappa_{A}(x)) = A\), then
  \[s((\Lambda \alpha)(\kappa_{A}(x))) = s(\kappa_{A}(\alpha(x))) \subseteq A = s(\kappa_{A}(x))\] and so \(s \cdot (\Lambda \alpha) \subseteq s\).

  In the other direction, the functor \(\Gamma \colon \Suppset \rightarrow \PShpb \Inc\) maps \(s \colon X \rightarrow \powf \V\) to the functor \(\Gamma X\) with \[\Gamma X A = \{x \in X \mid s(x) \subseteq A\}\] so that  \(\Gamma X (A \subseteq B)\) is just inclusion, obviously \(\Gamma X\) preserves  intersections. On morphisms \(\Gamma\) is equally trivial.

  These functors are indeed  pseudoinverses as manifested by the following isomorphisms:
  \begin{align*}
    F &\cong  \Gamma \Lambda F  & X &\cong \Lambda \Gamma X\\
    x \in FA &\mapsto \kappa_{A}(x) & x &\mapsto \kappa_{s(x)}(x)\\
    (s(y) \subseteq A ) \cdot y_{s} &\mapsfrom (\kappa_{B}(y) \in \Gamma (\Lambda F) A) & y&\mapsfrom \kappa_{A}(y)
  \end{align*}
          where $y_s \in s(y) \subseteq A$ is the smallest representative in the equivalence class of $\kappa_A(y)$.
  \end{enumerate}

\appendixproof{prop:suppset-shv}
  Let \(F \in \PSh \Inc\) be a presheaf.
  Then \(F\) is an \(\mathcal{O}\)-sheaf iff it preserves intersections of the form \(A \subseteq A \cup B \supseteq B\), and it is \(\mathcal{I}_{\Inc}\)-separated if every \(F(A \subseteq B) \colon FA \rightarrow FB\) is injective.
  In particular, if \(F\) preserves intersections it is a \(\mathcal{O}\)-sheaf and \(\mathcal{I}_{\Inc}\)-separated.

  For the converse direction, we assume that $F$ is a $\mathcal{O}$-sheaf and $\mathcal{I}_{\Inc}$-separted and show it preserves intersections:
  given subsets \(A, B \subseteq C\) with \(x \in FA, y \in FB\) such that \((A \subseteq C)\cdot x = (B \subseteq C) \cdot y\) we have to show there exists a unique \(z \in F(A \cap B)\) mapped to \(x\) and \(y\).
  We have \((A \subseteq A \cup B) \cdot x = (B \subseteq A \cup B)\cdot y\) since
  \begin{align*}
    &\; (A \cup B \subseteq C) \cdot (A \subseteq A \cup B)\cdot x \\
    =&\;(A \subseteq C) \cdot x \\
    =&\; (A \subseteq C) \cdot y \\
    =&\; (A \cup B \subseteq C)\cdot (B \subseteq A \cup B)\cdot y \\
  \end{align*}
  and \(F(A \cup B \subseteq C)\) is mono. This means that \(x, y\) is a compatible family for the \(\mathcal{O}\)-cover \(\{A, B\}\) of \(A \cap B\), and since \(F\) is a \(\mathcal{O}\)-sheaf this yields the desired unique element \(z \in F(A \cap B)\).

\end{document}